\documentclass[accepted]{uai2022} 

\usepackage[american]{babel}

\usepackage{natbib} 
    
\usepackage{mathtools} \usepackage{thmtools}
\usepackage{thm-restate}
\usepackage{booktabs} \usepackage{tikz} 

\usepackage[linesnumbered]{algorithm2e}

\usetikzlibrary{arrows,shapes,automata,backgrounds,petri,calc}
\usetikzlibrary{positioning,decorations.pathreplacing}

\tikzstyle{max}=[thick,draw,minimum size=1.4em,inner sep=0em]
\tikzstyle{min}=[diamond,thick,draw,minimum size=1.4em,inner sep=0em]
\tikzstyle{ran}=[circle,thick,draw,minimum size=1.4em,inner sep=0em]
\tikzstyle{target}=[rectangle, rounded corners, thick,draw,minimum size=1.4em,inner sep=0em]
\tikzstyle{act}=[circle,thick,draw,fill,minimum size=.7em,inner sep=0em]
\tikzstyle{mc}=[rounded corners,thick,draw,minimum size=1.4em,inner sep=.5ex]
\tikzstyle{tran}=[thick,draw,->,>=stealth]
\tikzstyle{loop left}=[tran, to path={.. controls +(150:.5)
	and +(210:.5) .. (\tikztotarget) \tikztonodes}]
\tikzstyle{loop right}=[tran, to path={.. controls +(30:.5)
	and +(330:.5) .. (\tikztotarget) \tikztonodes}]
\tikzstyle{loop above}=[tran, to path={.. controls +(60:.5)
	and +(120:.5) .. (\tikztotarget) \tikztonodes}]
\tikzstyle{loop below}=[tran, to path={.. controls +(240:.5)
	and +(300:.5) .. (\tikztotarget) \tikztonodes}]

\newcommand{\tm}{\mathit{time}}
\newcommand{\Nset}{\mathbb{N}}

\newcommand{\Rset}{\mathbb{R}}

\newcommand{\calA}{\mathcal{A}}

\newcommand{\calH}{\mathcal{H}}

\newcommand{\calV}{\mathcal{V}}

\newcommand{\wait}{\mathit{wait}}
\newcommand{\attack}{\mathit{attack}}

\newcommand{\EU}{\mathrm{EAU}}

\newcommand{\Val}{\operatorname{DVal}}

\newcommand{\AVal}{\operatorname{AVal}}

\newcommand{\Obs}{\Omega}

\newcommand{\NP}{\mathsf{NP}}

\newcommand{\PSPACE}{\mathsf{PSPACE}}
\newcommand{\Prob}{\operatorname{Prob}}
\newcommand{\Mem}{\mathit{mem}}

\newcommand{\Attack}{\mathit{Att}}

\newcommand{\hole}{\mathit{Hole}}

\newcommand{\ltime}{\ell}

\newcommand{\dhat}[1]{\smash{\hat{#1}}}
\newcommand{\dm}{d_{\max}}

\newcommand{\Steal}{\mathrm{Steal}}
\newcommand{\changesize}{\mathrm{CS}}

\newcommand\mem{\operatorname{mem}}
\let\hat\widehat

\usepackage{xspace}
\makeatletter
\DeclareRobustCommand\onedot{\futurelet\@let@token\@onedot}
\def\@onedot{\ifx\@let@token.\else.\null\fi\xspace}

\newcommand\eg{{e.g}\onedot} 
\newcommand\ie{{i.e}\onedot}

\makeatother

\newcommand{\Regstar}{\textsc{Regstar}\xspace}

\newcommand{\tran}[1]{\stackrel{\raisebox{-.3ex}{\scriptsize$#1$}}{\rightarrow}}

\def\ie{i.e\onedot}

\usepackage{amsfonts,amsmath,amsthm}
\theoremstyle{plain}
\newtheorem{theorem}{Theorem}

\newtheorem{remark}[theorem]{Remark}
\newtheorem{example}[theorem]{Example}

\usepackage[
	textsize=tiny,
	textwidth=1.8cm,
	colorinlistoftodos,
	backgroundcolor=teal!30!white,
	linecolor=teal!50!white
	]{todonotes}

\graphicspath{{graphics/}}

\usepackage{booktabs}
\usepackage{makecell}
\usepackage{multirow}
\usepackage{adjustbox}
\usepackage{hyperref}
\usepackage{tabularx}
\usepackage{collcell}
\usepackage{subcaption}
\usepackage{colortbl,xfp}
\usepackage{flushend}
\hypersetup{colorlinks=true, linkcolor=teal, citecolor=violet}
\clearpage{}\makeatletter
\def\DC@#1#2#3{\uccode`\~=`#1\relax
  \m@th
  \afterassignment\DC@x\count@#3\relax{#1}{#2}}
\def\DC@x#1\relax#2#3{\ifnum\z@>\count@
    \expandafter\DC@centre
  \else
    \expandafter\DC@right
  \fi
  {#2}{#3}{#1}}
\def\DC@centre#1#2#3{\let\DC@end\DC@endcentre
  \uppercase{\def~}{\egroup\setbox\tw@=\hbox\bgroup{#2}}\setbox\tw@=\hbox{{\phantom{{#2}}}}\setbox\z@=\hbox\bgroup\catcode`#1=13 }
\def\DC@endcentre{\egroup
    \ifdim \wd\z@>\wd\tw@
      \setbox\tw@=\hbox to\wd\z@{\unhbox\tw@\hfill}\else
      \setbox\z@=\hbox to\wd\tw@{\hfill\unhbox\z@}\fi
    \box\z@\box\tw@}
\def\DC@right#1#2#3{\ifx\relax#3\relax
    \hfill
    \let\DC@rl\bgroup
  \else
    \edef\DC@rl{to\the\count@\dimen@ii\bgroup\hss\hfill}\count@\@gobble#3\relax
  \fi
  \let\DC@end\DC@endright
  \uppercase{\def~}{\egroup\setbox\tw@\hbox to\dimen@\bgroup{#2}}\setbox\z@\hbox{1}\dimen@ii\wd\z@
   \dimen@\count@\dimen@ii
   \setbox\z@\hbox{{#2}}\advance\dimen@\wd\z@
   \setbox\tw@\hbox to\dimen@{}\setbox\z@\hbox\DC@rl\catcode`#1=13 }
\def\DC@endright{\hfil\egroup\box\z@\box\tw@}
\newcolumntype{D}[3]{>{\DC@{#1}{#2}{#3}}c<{\DC@end}}
\makeatother
\clearpage{}
\newcolumntype{d}[1]{D{+}{\,$\pm$\,}{#1}}

\title{On-the-fly Adaptation of Patrolling Strategies in Changing Environments}

\author[\space]{Tom\'a\v s Br\'azdil}
\author[\space]{David Kla\v ska} 
\author[\space]{Anton\'in Ku\v cera}
\author[\space]{V\'it Musil}
\author[\space]{Petr Novotn\'y}
\author[\space]{Vojt\v ech \v Reh\'ak}
\affil[\space]{Masaryk University\\
	Faculty of Informatics\\
	Brno, Czechia}

\begin{document}
\maketitle

\begin{abstract}
We consider the problem of efficient patrolling strategy adaptation in a changing
environment where the topology of Defender's moves and the importance of
guarded targets change unpredictably. The Defender must instantly switch to a
new strategy optimized for the new environment, not disrupting the ongoing
patrolling task, and the new strategy must be computed promptly under all
circumstances. Since strategy switching may cause unintended security
risks compromising the achieved protection, our solution includes mechanisms for detecting and mitigating this problem.
The efficiency of our framework is evaluated experimentally.
\end{abstract}

\section{Introduction}
\label{sec-intro}

In \emph{patrolling games}, a \emph{Defender} moves among vulnerable targets and strives to detect a possible ongoing attack.
The targets are modeled as vertices in a directed graph, where the edges correspond to admissible moves of the Defender.

An attack at a target $\tau$ takes $d(\tau)$ time units to complete successfully.
If an initiated attack is \emph{not} discovered in the next $d(\tau)$ time units, the Defender loses a utility determined by the cost of~$\tau$. The \emph{protection value} of a Defender's strategy $\sigma$ is the expected Defender's utility guaranteed by $\sigma$ against an arbitrary Attacker's strategy. 

\emph{Adversarial} patrolling assumes a powerful Attacker who can observe Defender's moves, know the Defender's strategy and use this information to identify the best attack opportunity.
The Defender's moving strategy is typically \emph{randomized}
\citep{KKMR:Regstar-UAI}
to prevent the Attacker from fully anticipating future moves.
The adversarial setting is particularly apt when the real Attacker's abilities are \emph{unknown} and certain protection degree is required even in the worst case.

Existing works focus on computing a Defender's strategy (moving plan) maximizing the protection value in a \emph{fixed} patrolling graph. This is challenging on its own because even special variants of the problem are $\PSPACE$-hard \citep{HO:UAV-problem-PSPACE}. However, having the underlying graph fixed is a significant limitation since
the environment \emph{does} change in real-life use cases over time and the Defender is required to adapt its strategy on the fly.
For instance, admissible moves of a police patrol are influenced by car accidents or traffic intensity, patrolling drones are affected by weather~etc.
The target costs also naturally evolve; for example, the cost of a storage place decreases when emptied,~etc.

\begin{figure}[t]
\centering
\begin{tikzpicture}[x=1cm,y=1cm,xscale =.8, yscale=.6,every node/.style={scale=.8}]
\draw[tran, very thick,->] (-5,0) -- (5,0);
\node[draw=none] at (5,-.4) {\textit{time}};
\filldraw[black] (0,0) circle (3pt) node[below=1ex]{$t$};
\filldraw[black] (-4.5,0) circle (2pt) node[below=1ex]{$t_0$};
\filldraw[black] (-2.7,0) circle (2pt) node[below=1ex]{$t_0{+}d(\tau)$};
\filldraw[black] (-.8,0) circle (2pt) node[below=1ex]{$t_1$};
\filldraw[black] (1,0) circle (2pt) node[below=1ex]{$t_1{+}d(\tau)$};
\filldraw[black] (2,0) circle (2pt) node[below=1ex]{$t_2$};
\filldraw[black] (3.8
,0) circle (2pt) node[below=1ex]{$t_2{+}d(\tau)$};
\draw[thick,dotted,rounded corners] (-4.5,0) -- +(0,.4) -- node[above] {\textit{attacking } $\tau$}  +(1.8,0.4) -- +(1.8,0);
\draw[thick,dotted,rounded corners] (-.8,0) -- +(0,.4) -- node[above] {\textit{attacking } $\tau$}  +(1.8,0.4) -- +(1.8,0);
\draw[thick,dotted,rounded corners] (2,0) -- +(0,.4) -- node[above] {\textit{attacking } $\tau$}  +(1.8,0.4) -- +(1.8,0);
\draw [decorate,decoration={brace,amplitude=5pt},xshift=-1pt,yshift=0pt] (0,-1) -- (-5,-1) node [black,midway,yshift=-0.6cm] {\textit{Defender uses $\sigma_1$ in $G_1$}};
\draw [decorate,decoration={brace,amplitude=5pt},xshift=1pt,yshift=0pt] (5,-1) -- (0,-1) node [black,midway,yshift=-0.6cm] {\textit{Defender uses $\sigma_2$ in $G_2$}};
\end{tikzpicture}
\caption{
The coverage of the attack initiated at time $t_1$ short before the strategy switch can be very low due to the ``incompatibility'' of strategies $\sigma_1$ and $\sigma_2$.} 
\label{fig-sec-hole}
\end{figure}
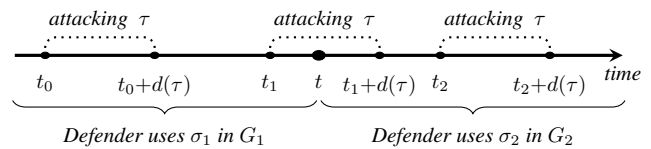

When the patrolling graph $G_1$ changes into $G_2$, the current Defender's strategy $\sigma_1$ must be promptly replaced with another strategy $\sigma_2$ optimized for the new graph. In principle, $\sigma_2$ can be computed by one of the existing strategy synthesis algorithms for \emph{fixed} patrolling graphs with $G_2$ on input. However, we show that this approach has a major \emph{conceptual flaw}. 
Namely, ignoring the functionality of $\sigma_1$ when constructing $\sigma_2$ may lead to creating unnecessary \emph{security holes} caused by the ``incompatibility'' of $\sigma_1$ and~$\sigma_2$. 
Furthermore, existing algorithms for fixed patrolling graphs are \emph{not} sufficiently efficient to be run under real-time constraints.

To understand the origin and impact of security holes, consider the scenario of Fig.~\ref{fig-sec-hole}.
Here a patrolling graph $G_1$ changes into $G_2$ at time $t$, and a Defender's strategy $\sigma_1$ is replaced with $\sigma_2$.
Since $\sigma_1$ and $\sigma_2$ are optimized for $G_1$ and $G_2$ respectively, they plan visits to all targets (including $\tau$) so that the expected damage is constrained by the protection values of $\sigma_1$ and $\sigma_2$.
An attack at $\tau$ initiated at time $t_0\le t-d(\tau)$ is fully covered by $\sigma_1$, and an attack initiated at time $t_2\geq t$ is fully covered by $\sigma_2$.
Hence, these attacks are no more dangerous than others.
Now consider an attack initiated at time $t_1$ ``short before'' the strategy switch.
If $\sigma_2$ ignores the functionality of $\sigma_1$, it may happen that   
$\sigma_1$ does not patrol $\tau$ in the first $t-t_1$ time units, and $\sigma_2$ omits $\tau$ in the next $d(\tau)-(t-t_1)$ time units (\ie, both strategies plan to visit $\tau$  ``later'').
If this happens, switching from $\sigma_1$ to $\sigma_2$ at time~$t$ creates a temporary but \emph{exceptionally dangerous} attack opportunity, \ie, a \emph{security hole}.
A simple concrete instance with quantitative analysis is given in Example~\ref{exa-change}.  

Large security holes are particularly awkward when environmental changes are frequent. Regardless of their frequency, 
security holes compromise the protection quality and cannot be ignored when we aim at providing  robust security guarantees under all circumstances.
In general, the difference between $G_1$ and $G_2$ may be so large that creating security holes becomes \emph{unavoidable} (see Example~\ref{exa-unavoidable}).
This motivates the problems of algorithmic \emph{detection}, \emph{analysis}, and \emph{mitigation} of security holes for a given pair of strategies~$\sigma_1$~and~$\sigma_2$.

The term ``Defender'' actually refers to the whole patrolling infrastructure, including systems for observing environmental changes, synthesizing new strategies, and deploying them to the moving agents.
Hence, we assume the Defender observes environmental changes when they happen, and it has sufficient computational resources at its disposal. For the Attacker, we keep the \emph{worst-case} approach, assuming it can observe Defender's moves, environmental changes when they happen and knows the Defender's strategies before/after the change. Furthermore, when evaluating the achieved protection, we assume the environment changes at the moment \emph{least convenient} for the Defender. Consequently, the constructed strategies are \emph{resistant even to sophisticated attacks when the Attacker utilizes all of this information}.

\paragraph{Contribution}

\emph{We efficiently solve the problem of on-the-fly patrolling strategy adaptation in a changing adversarial environment}.
Our approach overcomes the aforementioned problems and is applicable to real-world scenarios. Namely:
\begin{itemize}
\item[(1)] We introduce an appropriate formal model for changing environments and strategy switching.
\item[(2)] We formalize the concept of a security hole. We design an efficient algorithm for detecting and estimating security holes caused by a given strategy switch. 
\item[(3)] We design an algorithm for computing a Defender's strategy $\sigma_2$ replacing the original strategy $\sigma_1$ when the underlying patrolling graph $G_1$ changes into $G_2$. This algorithm reduces the danger of creating large security holes and it is \emph{sufficiently efficient} to be run on the fly.
\item[(4)] We show that, under certain conditions, security holes can be \emph{mitigated} by randomized strategy switching.
\item[(5)] We confirm the efficiency of our algorithms experimentally on instances of considerable size.
\end{itemize}
As a byproduct of our effort, we obtain a strategy synthesis algorithm for fixed patrolling graphs outperforming the best existing algorithm by a margin.

Existing works on patrolling in dynamic environments are applicable to special graph topologies, non-adversarial environment, or concentrate on collaborative problems such as optimal reassigning the targets to agents (see Related Work). To the best of our knowledge, the presented results are the \emph{first attempt} to solve the problem of \emph{dynamic adaptation of moving strategies in adversarial changing environment with general topology}. We believe that the introduced concept of a security hole is of broader interest. The underlying observations may help to handle similar issues in a larger class of dynamic planning problems with recurrent time-bounded objectives, where the new strategy is obliged to satisfy the commitments not fully accomplished by the old strategy.

\subsection{Related Work}

Our paper fits the \emph{security games} line of work studying optimal allocation of limited security resources for achieving optimal target coverage \citep{Tambe:book}.
Practical applications of security games include the deployment of police checkpoints at the Los Angeles International Airport \citep{PJMOPTWPK:Deployed-ARMOR}, the scheduling of federal air marshals over the U.S.{} domestic airline flights \citep{TRKOT:IRIS}, the arrangement of city guards in Los Angeles Metro \citep{DJYZTKS:patrolling-uncertainty-JAIR}, the positioning of U.S.{} Coast Guard patrols to secure selected locations \citep{ASYTBDMM:Protect-AIMagazine}, and also applications to wildlife protection in Uganda~\citep{FKDYT:PAWS}. 

Most of the previous results about \emph{adversarial patrolling games} where the Defender is mobile, the environment is actively hostile, and the game horizon is infinite concentrate on computing an optimal moving strategy for certain graph topologies.
The underlying solution concept is the \emph{Stackelberg equilibrium} \citep{SFAKT:Stackelberg-Security-Games,YKKCT:Stackelberg-Nash-security}, where the Defender/Attacker play the role of the \mbox{Leader/Follower}. 

For general topologies, the existence of a perfect Defender's strategy discovering all attacks in time is $\PSPACE$-complete \citep{HO:UAV-problem-PSPACE}.
Consequently, computing an optimal Defender's strategy is $\PSPACE$-hard.
Moreover, computing an $\varepsilon$-optimal strategy for $\varepsilon \leq 1/2n$, where $n$ is the number of vertices, is $\NP$-hard \citep{KKR:patrol-drones}.
Hence, no feasible strategy synthesis algorithm can \emph{guarantee} (sub)optimality for all inputs, and finding high-quality strategy in reasonable time is challenging.
The existing methods are based on mathematical programming, reinforcement learning, or gradient descent.
The first approach suffers from scalability issues caused by non-linear constraints \citep{BGA:large-patrol-AI,BGA:patrolling-arbitrary-topologies}.
Reinforcement learning has so far been successful mainly for patrolling with finite horizon, such as green security games \citep{WSYWSJF:Patrolling-learning,BAVT:learn-preventive-healthcare,Xu:Green-security,KMZGA:moving_targets}. Gradient descent techniques for finite-memory strategies \citep{KL:patrol-regular,KKLR:patrol-gradient,KKMR:Regstar-UAI} are applicable to patrolling graphs of reasonable size. Strategy synthesis for restricted topologies has been studied for lines, circles  \citep{AKK:multi-robot-perimeter-adversarial,ASKK:perimeter-patrol}, or fully connected environments \citep{BKR:patrol-Internet}.

Dynamically changing environments have so far been considered mainly in the context of multi-agent patrolling where the task is to dynamically reassign the targets to agents \citep{Multi-agent-dynamic,Dynamic-Environment-Patrolling,Decentralized-Patrolling,JRM:dynamic-patrolling,dynamic-patrol-ring}.

\section{Background}

We recall the standard notions of a patrolling graph, Defender's and Attacker's strategies and their values.
Since our experiments also involve comparison with state-of-the-art strategy synthesis algorithm for fixed patrolling graphs \citep{KKMR:Regstar-UAI}, we adopt the same setup.

\paragraph{Patrolling graph}
A (static) \emph{patrolling graph} is a tuple $G = (V,T,E,\tm,d,\alpha)$ where
\begin{itemize}
	\item $V$ is a finite set of \emph{vertices} (Defender's positions);
	\item $T \subseteq V$ is a non-empty set of \emph{targets}; 
	\item $E \subseteq V \times V$ is a set of \emph{edges} (admissible moves);
	\item $\tm\colon E\to\Nset_+$ specifies the time to travel an edge;
	\item $d\colon T\to\Nset_+$ assigns the time to complete an attack;
	\item $\alpha\colon T\to\Rset_+$ defines the costs of targets.
\end{itemize}

We write \mbox{$u\to v$} instead of $(u,v)\in E$, and denote $\alpha_{\max}=\max_{\tau\in T} \alpha(\tau)$ and $\dm=\max_{t\in T} d(t)$. In the sequel, let $G$ be a fixed patrolling graph.

\paragraph{Defender's strategy}
\label{sec-Defender-strategy}

In general, the Defender may choose the next vertex randomly depending on the whole history of previously visited vertices. As observed by \citet{KKMR:Regstar-UAI}, a subclass of \emph{regular} Defender's strategies achieves the same limit protection as general strategies, and it is more convenient for algorithmic synthesis. 

In the area of graph games, regular strategies are also known as \emph{finite-memory strategies with stochastic memory update}. Intuitively, such a strategy is represented by a finite-state probabilistic automaton $\calA$ that ``reads'' the sequence of vertices visited so far. When a new vertex~$v$ is read, $\calA$ changes its current state $m$ into another state $m'$ chosen randomly according to a fixed probability distribution determined by $m$ and $v$. The decision taken by the strategy then depends only on the vertex currently visited and the current state of $\calA$. Hence, the set of states of $\calA$, denoted by $\Mem$, can be seen as a finite memory where some information about the history of visited vertices is stored (we also refer to the states of $\calA$ as \emph{memory elements}). 

Formally, let $\Mem$ be a finite set. The corresponding set of \emph{augmented vertices} $\dhat V$ is defined as $V \times \Mem$, and we use $\hat{v}$ to denote an augmented vertex of the form $(v,m)$. An \emph{augmented edge} is a pair $\hat{e} \equiv (\hat{v},\hat{u})$ of augmented vertices where $e \equiv (v,u) \in E$. The set of all augmented edges is denoted by $\dhat{E}$.

A \emph{regular Defender's strategy} for $G$ is a function $\sigma$ assigning to every $\hat{v} \in \dhat{V}$ a probability distribution over $\dhat{V}$ so that  $\sigma(\hat{v})(\hat{u})>0$ only if $v\to u$. Intuitively, the Defender starts in some $v \in V$ where the state of $\calA$ is initialized to some $m \in \Mem$, and then it randomly selects the next vertex and the next memory element according to $\sigma$. Thus, $\sigma$~encodes both the selection of the next vertex and the choice of the next state performed by $\calA$ (there is no need to specify the transitions of $\calA$ explicitly). 

Let us fix an initial augmented vertex $\hat{v}$. For every finite sequence $h = \hat{v}_1,\ldots,\hat{v}_n$, we use $\Prob(h)$ to denote the probability of executing $h$ under $\sigma$ when the Defender starts patrolling in $\hat{v}$. That is, \mbox{$\Prob(h) = 0$} if $\hat{v} \neq \hat{v}_1$, otherwise $\Prob(h) = \prod_{i=1}^{n-1} \sigma(\hat{v}_i)(\hat{v}_{i+1})$. Whenever we write $\Prob(h)$, the associated $\sigma$ and $\hat{v}$ are clearly determined by the context.

\paragraph{Attacker's strategy}

In the patrolling graph, the time is spent by traversing edges. Adversarial patrolling assumes a powerful  Attacker capable of determining the next edge taken by the Defender immediately after its departure from the vertex currently visited. For the Attacker, this is an optimal moment to attack because delaying the attack gains no advantage (as we shall see, this is no longer true in a \emph{changing} environment). Furthermore, the 
Attacker can attack \emph{at most once} during a play.

An \emph{observation} is a sequence $o =
v_1,\ldots, v_n, v_n {\rightarrow} v_{n+1}$, where $v_1,\ldots, v_n$ is a path in~$G$. Intuitively, $v_1,\ldots, v_n$ is the sequence of vertices visited by the Defender, $v_n$ is the currently visited vertex, and $v_n {\rightarrow} v_{n+1}$ is the edge taken next. 
The set of all observations is denoted by $\Obs$. 
An \emph{Attacker's strategy} is a function
$\pi\colon \Obs \rightarrow \{\wait,\attack_\tau:\tau\in T\}$. As usual, we require that
if $\pi(v_1,\ldots, v_n,v_n {\rightarrow} u) = \attack_\tau$ for some $\tau \in T$,
then $\pi(v_1,\ldots, v_i,v_i {\rightarrow} v_{i+1}) = \wait$ for all $1\leq
i<n$. Intuitively, this ensures that the Attacker can attack at most once (this assumption is standard; see, e.g., \citep{KKLR:patrol-gradient,KKMR:Regstar-UAI} for a more detailed explanation).

\paragraph{Evaluating Defender's strategy}

\newcommand\Eval{\operatorname{Eval}}

Let $\sigma$ be a regular Defender's strategy and $\pi$ an Attacker's strategy. 

Let us fix an initial augmented vertex $\hat{v}$ where the Defender starts patrolling. The \emph{expected Attacker's utility} for $\sigma$, $\pi$ and $\hat{v}$ is defined as
\begin{equation*}
	\EU^{\sigma,\pi}(\hat{v})
	\ = \  \sum_{\tau,\hat{e}}  \mathbf{P}^{\sigma,\pi}(\hat{e},\tau) \cdot \Steal^\sigma(\hat{e},\tau)
\end{equation*}
where $\mathbf{P}^{\sigma,\pi}(\hat{e},\tau)$ is the probability of initiating an attack at $\tau$ when the  Defender starts moving along~$\hat{e}$, and $\Steal^\sigma(\hat{e},\tau)$ is the expected cost ``stolen'' by this attack. 

More precisely, let $\Attack(\pi,\hat{e},\tau)$ be the set of all $(\hat{v}_1,\ldots,\hat{v}_{n+1})$ such that $\pi(v_1,\ldots,v_n,v_n {\to} v_{n+1}) = \tau$ and $\hat{e} = \hat{v}_n \to \hat{v}_{n+1}$. We put
\begin{equation*}
   \mathbf{P}^{\sigma,\pi}(\hat{e},\tau) = \sum_{h \in \Attack(\pi,\hat{e},\tau)} \Prob(h).
\end{equation*}
Furthermore, let $\mathbf{M}^{\sigma}(\hat{e},\tau)$ be the probability of missing (i.e., not visiting) an augmented vertex of the form $\hat{\tau}$ in the first $d(\tau) - \tm(e)$ time units by a Defender's walk initiated in $\hat{u}$, where $\hat{u}$ is the destination of $\hat{e}$. We define 
$\Steal^\sigma(\hat{e},\tau) = \alpha(\tau) \cdot \mathbf{M}^{\sigma}(\hat{e},\tau)$.

Intuitively, $\EU^{\sigma,\pi}(\hat{v})$ is the expected amount ``stolen'' by the Attacker.
The Defender and Attacker aim to minimize and maximize $\EU^{\sigma,\pi}(\hat{v})$, respectively. The \emph{Attacker's value of $\sigma$ in $\hat{v}$} is the expected Attacker's utility achievable when the Defender commits to~$\sigma$ and starts patrolling in $\hat{v}$, i.e., 
$\AVal_G(\sigma)(\hat{v})  =  \sup_{\pi} \ \EU^{\sigma,\pi}(\hat{v})$.
The Defender can choose the initial $\hat{v}$, and hence we also define the \emph{Attacker's value of $\sigma$} as 
\begin{equation*} \label{E:steal-def}
	\AVal_G(\sigma) =  \min_{\hat{v}} \AVal_G(\sigma)(\hat{v}).
\end{equation*}
The \emph{Defender's value} (or simply the \emph{value}) is defined by
\begin{align*}
   \Val_G(\sigma)(\hat{v}) & = \alpha_{\max} - \AVal_G(\sigma)(\hat{v})\\
   \Val_G(\sigma)      & = \alpha_{\max} - \AVal_G(\sigma).
\end{align*}
Intuitively, $\Val_G(\sigma)$ corresponds to the protection guaranteed by $\sigma$ against an arbitrary Attacker's strategy.
We omit the `$G$' subscript if it is clear from the context.

\section{Changing Environment}

In this section, we introduce a formal model of changing environments, formalize the concept of strategy switching, and show how to evaluate a switching strategy in a changing environment. 

We consider two types of environmental changes: \emph{topological changes} influencing the admissible Defender's moves, \ie, inserting/deleting edges or modifying edge traversal time, and \emph{utility changes} modifying the targets costs.

Formally, a \emph{changing environment} is a pair $G_1 \mapsto G_2$ where $G_1$ and $G_2$ are patrolling graphs with the same set of vertices $V$, the same set of targets $T$, and the same $d$ specifying the attack times. We write $E_i$, $\tm_i$, and $\alpha_i$ to denote the edges, traversal times, and target costs of $G_i$ for $i \in \{1,2\}$.

Note that our definition does not allow changing the vertex set or the target set, yet
these changes can be easily modeled. For instance, adding a vertex may be modeled such that the vertex is present in both $G_1$ and $G_2$ but has no incoming edges in $G_1$ ($\sigma_1$ will be extended with an arbitrary behavior at the vertex). Similarly, removing a target may be modeled by changing its cost to a negligibly small value.

For the rest of this section, we fix a changing environment $G_1 \mapsto G_2$, and a pair of regular Defender's strategies $\sigma_1$ and $\sigma_2$ for $G_1$ and $G_2$, respectively. We assume that $\sigma_1$ and $\sigma_2$ use the same set $\Mem$ of memory elements.

\paragraph{Strategy switching}
 
Let $t \in \Nset$ be a \emph{switching time}. We use $G_1 \mapsto_t G_2$ to denote the scenario where the patrolling graph $G_1$ changes into the patrolling graph $G_2$ at time~$t$, and $\sigma_1 \mapsto_t \sigma_2$ to denote the Defender's strategy for \mbox{$G_1 \mapsto_t G_2$} obtained by ``switching'' from $\sigma_1$ into $\sigma_2$ at time~$t$, defined as follows. 

The Defender keeps executing $\sigma_1$ in all augmented vertices visited strictly before time~$t$. Let $(v,m)$ be the first augmented vertex visited by the Defender at or after time $t$ (observe that the $m$ is still determined by $\sigma_1$). From now on, the Defender should play according to $\sigma_2$. We distinguish three possibilities.

\begin{itemize}
\item[(a)] There is $m' \in \Mem$ such that $\Val_{G_2}(\sigma_2) = \Val_{G_2}(\sigma_2)(v,m')$. 
Then, the Defender selects such an $m'$ and starts applying $\sigma_2$ from $(v,m')$.
\item[(b)] The condition of (a) does not hold, but there exist $(v',m')$ and a path from $v$ to $v'$ in $G_2$ such that $\Val_{G_2}(\sigma_2) = \Val_{G_2}(\sigma_2)(v',m')$. Then, the strategy $\sigma_1 \mapsto_t \sigma_2$ follows the selected path from $v$ to~$v'$, and then starts applying $\sigma_2$ from $(v',m')$ for the selected $m'$.   
\item[(c)] None of the conditions (a) and (b) holds. Then, it is \emph{impossible} to perform a switch from $\sigma_1$ to $\sigma_2$ preserving the protection value of $\sigma_2$, and the strategy $\sigma_1 \mapsto_t \sigma_2$ is undefined.
\end{itemize}
In all scenarios considered in our experiments, Condition~(a) holds for every~$t$. Condition~(b) corresponds to a situation when some  vertex $v$ visited by $\sigma_1$ is no longer visited by $\sigma_2$. Condition~(c) covers pathological cases when a ``drastic'' environmental change prevents switching $\sigma_1$ into $\sigma_2$ (e.g., all edges disappear). From now on, we assume that Condition~(a) or~(b) holds and the strategy $\sigma_1 \mapsto_t \sigma_2$ is defined.

\begin{remark}
\label{rem-compatibility}
Our algorithm for constructing $\sigma_2$ (see Preventing and Mitigating Security Holes) ``adapts'' $\sigma_1$ to the new environment $G_2$. Hence, the elements of $\Mem$ may represent similar information about the history of visited vertices in $\sigma_1$ and $\sigma_2$, and Condition~(a) may hold even for $m'=m$. In this case, the information encoded by $m$ is passed on to $\sigma_2$ during the switch, decreasing the danger of creating large security holes.  
\end{remark}

\paragraph{Evaluating a switching strategy} 

The notions defined for static environments (Attacker's strategy, expected utility, strategy value, etc.) also apply to changing environments, and the technical adjustments are trivial. However, the notion of Attacker's observation requires revision for the reasons described below.

In static scenarios, it is safe to assume the Attacker initiates his attack when the Defender leaves a vertex (see the paragraph \emph{Attacker's strategy} in the previous section). However, in $G_1 \mapsto_t G_2$, the Attacker \emph{may} increase its expected utility by initiating an attack in the middle of a Defender's move. This is because a short delay may suffice for completing the attack \emph{after} time $t$ when the target becomes more valuable, but postponing the attack to the moment when the Defender completes the move would already increase the probability of discovering the attack too much. In case of deeper interest, see a concrete example in Appendix~\ref{app-changing}.

Technically, we define an \emph{observation} in  $G_1 \mapsto_t G_2$ as a pair $(o,\delta)$, where $o = v_1,\ldots, v_n,v_n {\rightarrow} v_{n+1}$ is defined as for static environments and $\delta \in \Nset$ is a \emph{delay} strictly smaller than $\tm_i(v_n {\rightarrow} v_{n+1})$, where $i=1$ if the move $v_n {\rightarrow} v_{n+1}$ is initiated before time~$t$, and $i=2$ otherwise. 

The expected Attacker's utility $\EU^{\sigma_1 \mapsto_t \sigma_2,\pi}(\hat{v})$ is defined similarly as for static environments, i.e., as a sum
\begin{equation}
\label{eq-EAU-changing}
   	 \sum_{\tau,\hat{e},\delta,t_0}  \mathbf{P}^{\sigma_1 \mapsto_t \sigma_2,\pi}(\hat{e},\tau,\delta,t_0) \cdot \Steal^{\sigma_1 \mapsto_t \sigma_2}(\hat{e},\tau,\delta,t_0).
\end{equation}
Here, $t_0 \in \Nset$ denotes the attack time. The symbol $\mathbf{P}^{\sigma_1 \mapsto_t \sigma_2,\pi}(\hat{e},\tau,\delta,t_0)$ is the probability of initiating an attack at $\tau$ at time $t_0$ when the Defender has been going along~$\hat{e}$ for $\delta$ time units (note that this also depends on the Defender's initial position $\dhat{v}$). $\Steal^{\sigma_1 \mapsto_t \sigma_2}(\hat{e},\tau,\delta,t_0)$ denotes the expected cost ``stolen'' by this attack. Detailed technical definitions are in Appendix~\ref{app-changing}. Although the delay $\delta$ further complicates our technical definitions, it describes a \emph{real phenomenon} which must be properly reflected by a realistic formal model.

The time $t$ when $G_1$ changes into $G_2$ is unpredictable, and the strategy $\sigma_2$ must guarantee a reasonable protection on $G_2$ for all $t$'s. Hence, the \emph{Attacker's value of $\sigma_1 \mapsto \sigma_2$ in $G_1 \mapsto G_2$} is defined as
\begin{equation*}
    \AVal_{G_1 \mapsto G_2}(\sigma_1 {\mapsto} \sigma_2) =  \min_{\hat{v}}\ \sup_{\pi}\ \sup_{t}\
    \EU^{\sigma_1 \mapsto_t \sigma_2,\pi}(\hat{v})\,.
\end{equation*}

Note that the ``$\sup_\pi$'' in the above definition ensures that \emph{all} Attacker's strategies are taken into account, including those taking advantage of observing the environmental change, Defender's moves, and analyzing the functionality of $\sigma_1,\sigma_2$.

\section{Security Holes}

Strategy switching may result in temporarily decreasing the protection of some targets. Clearly, the Defender cannot protect $G_1 \mapsto G_2$ by  $\sigma_1 \mapsto \sigma_2$ better than it is protecting $G_1$ by $\sigma_1$ and $G_2$ by $\sigma_2$. In terms of Attacker's values,
\begin{eqnarray*}
   \AVal_{G_1 \mapsto G_2}(\sigma_1 \mapsto \sigma_2)  \ \geq \ \AVal_{G_1}(\sigma_1),\\
   \AVal_{G_1 \mapsto G_2}(\sigma_1 \mapsto \sigma_2)  \ \geq \ \AVal_{G_2}(\sigma_2).
\end{eqnarray*}
The first inequality is simple because the switching time~$t$ can be arbitrarily large. With increasing $t$, the Attacker can perform more and more of his attacks scheduled by a given strategy $\pi$ against $\sigma_1$ in $G_1$ also in $G_1 \mapsto G_2$, achieving the expected utility arbitrarily close to the expected utility received in $G_1$. The second inequality is also immediate because the Attacker can ``simulate'' an arbitrary strategy $\pi$ against $\sigma_2$ also in $G_1 \mapsto G_2$ by performing his attacks after the switching time.

However, it may also happen that $\AVal_{G_1 \mapsto G_2}(\sigma_1 {\mapsto} \sigma_2)$ is \emph{strictly larger} than the maximum of $\AVal_{G_1}(\sigma_1)$ and $\AVal_{G_2}(\sigma_2)$ due to the new attack  opportunities offered in the limited time window short before the switching time caused by the ``incompatibility'' between $\sigma_1$ and $\sigma_2$ (see Fig.~\ref{fig-sec-hole}). Note that although the Attacker cannot enforce an environmental change at a particular time, in our adversarial setting we consider the worst possibility, i.e., we assume the change happens in the least convenient moment. Formally, the \emph{security hole} of $\sigma_1 \mapsto \sigma_2$, denoted by $\hole_{G_1 \mapsto G_2}(\sigma_1,\sigma_2)$, is defined as
\begin{equation*}
  \AVal_{G_1 \mapsto G_2}(\sigma_1 \mapsto \sigma_2) - \max\{\AVal_{G_1}(\sigma_1), \AVal_{G_2}(\sigma_2)\}.
\end{equation*}
Intuitively, $\hole_{G_1 \mapsto G_2}(\sigma_1,\sigma_2)$ is the extra amount stolen by the Attacker due to the incompatibility between $\sigma_1$ and $\sigma_2$. Note that if $\hole_{G_1 \mapsto G_2}(\sigma_1,\sigma_2) = 0$, then the new attack opportunities caused by the switch are no more dangerous than the ones offered by $\sigma_1$ in $G_1$ and $\sigma_2$ in $G_2$.

\begin{example}
\label{exa-change}
Let $G_1$ and $G_2$ be the patrolling graphs of Fig.~\ref{fig-hole}.
Let $\sigma_1$ be a trivial strategy walking among $v_1,v_2,v_3$ \emph{clockwise}. Since every target is revisited within the next~$6$~time units, all attacks are discovered in time and hence $\AVal_{G_1}(\sigma_1) = 0$. When the environment changes into $G_2$ by removing the edge $v_2 \tran{} v_3$, the Defender's strategy is changed into $\sigma_2$ walking among $v_1,v_2,v_3$ \emph{anticlockwise}. Clearly, $\AVal_{G_2}(\sigma_2) = 0$, and hence both $\sigma_1$ and $\sigma_2$ achieve perfect protection in $G_1$ and $G_2$, resp.

Now consider the scenario where $v_3$ is attacked at time~$\ell$ when the Defender is in the middle of the move $v_3 \tran{} v_1$ in $G_1$. The Defender arrives in $v_1$ at time $\ell + 1$, and the environment changes from $G_1$ into $G_2$ at time $t = \ell+2$.  In $v_1$, the Defender still uses $\sigma_1$ to determine the next move, and arrives in $v_2$ at time $\ell + 3$. In $v_2$, the Defender already uses the new strategy $\sigma_2$, and therefore visits $v_3$ at time $\ell + 7$. That is, the attack at $v_3$ initiated at time $\ell$ (short before the switching time) \emph{succeeds} with probability one. Consequently, $\AVal_{G_1 \mapsto G_2}(\sigma_1 \mapsto \sigma_2) = 100$ and hence
$\hole_{G_1 \mapsto G_2}(\sigma_1,\sigma_2) = 100$. 
\end{example}

\begin{figure}[t]
\centering
\begin{tikzpicture}[x=1.7cm,y=1.7cm,xscale=0.9,yscale=0.5,every node/.style={scale=.8}]
			\node (v1) at (0,0)      [ran]  {$v_1$}; 
			\node (v2) at (-.57,-1)  [ran]  {$v_3$}; 
			\node (v3) at (0.57,-1)  [ran]  {$v_2$}; 
            \path[tran] (v1)  edge   [bend left=15] (v2);
            \path[tran] (v2)  edge   [bend left=15] (v1);
            \path[tran] (v2)  edge   [bend left=15] (v3);
            \path[tran] (v3)  edge   [bend left=15] (v2);
            \path[tran] (v1)  edge   [bend left=15] (v3);
            \path[tran] (v3)  edge   [bend left=15] (v1);
            \node at (0,-1.5) {$G_1$; $\sigma_1$ walks clockwise.};
\node (u1) at (3,0)      [ran]  {$v_1$}; 
			\node (u2) at (2.43,-1)  [ran]  {$v_3$}; 
			\node (u3) at (3.57,-1)  [ran]  {$v_2$}; 
            \path[tran] (u1)  edge   [bend left=15] (u2);
            \path[tran] (u2)  edge   [bend left=15] (u1);
\path[tran] (u2)  edge   [bend left=15] (u3);
            \path[tran] (u1)  edge   [bend left=15] (u3);
            \path[tran] (u3)  edge   [bend left=15] (u1);
            \node at (3,-1.5) {$G_2$; $\sigma_2$ walks anti-clockwise.};
\end{tikzpicture}
	\caption{Exemplifying $\hole_{G_1 \mapsto G_2}(\sigma_1,\sigma_2) = 100$. Target costs are $100$, traversing every edge takes $2$ time units, completing an attack takes $6$ time units.
	If the Attacker attacks $v_3$ when the Defender is in the middle of the edge $v_3\to v_1$ and the environment changes in $2$ more time units, the attack succeeds with probability $1$.
}
	\label{fig-hole}
\end{figure}
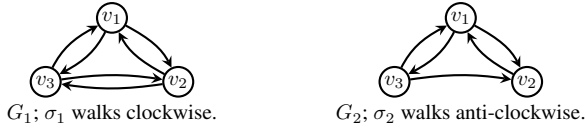

\subsection{Estimating security holes}
Now we present an algorithm for computing an upper bound on the security hole.
First, we reduce the estimation of the security hole to computation of certain steals. Second, we make some observations that allow us to consider as few of the steals as possible. Third, we present Algorithm~\ref{alg-search}, which manages to merge the computation
of several steals into one, by performing a search through the patrolling graph.

\paragraph{Reduction to steals}
Let $G_1,G_2$ be patrolling graphs and $\sigma_1,\sigma_2$ be
Defender's strategies in $G_1$ and $G_2$, respectively.
Recall that $\hole_{G_1 \mapsto G_2}(\sigma_1,\sigma_2)$ is defined as
\begin{equation*}
\AVal_{G_1 \mapsto G_2}(\sigma_1 \mapsto \sigma_2) - \max\{\AVal_{G_1}(\sigma_1), \AVal_{G_2}(\sigma_2)\}
\end{equation*}
Since $\AVal_{G_1}(\sigma_1)$ and $\AVal_{G_2}(\sigma_2)$ are computable by the standard strategy evaluation algorithm \cite[see, e.g.,][]{KKMR:Regstar-UAI}, we only need to compute an \emph{upper bound on $\AVal_{G_1 \mapsto G_2}(\sigma_1 {\mapsto} \sigma_2)$.} Recall that the expected protection achieved by $\sigma_1 \mapsto \sigma_2$ against a given attack is fully determined by the following:

\begin{itemize}
	\item $\dhat{e}$: the Defender's location when the attack is initiated;
	\item $\tau$: the attacked target;
	\item $\delta$: the time passed since the Defender entered $\dhat{e}$;
	\item $t_0$: the current time (when the attack is initiated);
	\item $t$: the switching time.
\end{itemize}

According to the definition of security hole, it suffices to compute the maximum of all the values $\Steal^{\sigma_1 \mapsto_t \sigma_2}(\hat{e},\tau,\delta,t_0)$. Since there are \emph{infinitely many} $t,t_0 \in \Nset$, this task is not trivial.

\paragraph{Minimizing the number of steals to consider}
If $t_0 < t {-} d(\tau)$ or  $t_0 \geq t$, then the attack at $\tau$ is fully covered by $\sigma_1$ or $\sigma_2$, respectively. Hence, the only interesting case is when \mbox{$1 \leq t - t_0 \leq d(\tau)$}. Although there are still infinitely many $t,t_0$ satisfying this condition, the above $\Steal$ is \emph{fully determined} just by the difference $t-t_0$, and can thus be written as $\Steal(\dhat{e},\tau,\delta,\Delta t)$, where $\Delta t$ denotes the difference and ranges over \emph{finitely many values} bounded by $d(\tau)$.

Another simple observation is that $\Steal(\dhat{e},\tau,\delta,\Delta t) \leq \Steal(\dhat{e},\tau,0,{\rm min}(\Delta t+\delta,d(\tau)))$.
Hence, from now on, we omit the $\delta$, implicitly assuming $\delta=0$.

Furthermore, for all $\dhat{e},\dhat{e}'$ such that $\tm_1(e)\leq \tm_1(e')$
and both $\dhat{e}$ and $\dhat{e}'$ lead to the \emph{same} augmented vertex, we have that
$\Steal(\dhat{e},\tau,\Delta t) \leq \Steal(\dhat{e}',\tau,\Delta t)$.
Therefore, it suffices to pick, for each augmented vertex $\dhat{v}$, one of the longest augmented edges leading to $\dhat{v}$, and disregard all other augmented edges when looking for the maximal $\Steal(\dhat{e},\tau,\Delta t)$.

Finally, for given $\dhat{e}$, we say that $\Delta t$
is an \emph{arrival time} if the Defender can reach some vertex in \emph{precisely}
$\Delta t$ time units after it starts moving along $\dhat{e}$. Note that if $\Delta t<d(\tau)$ is \emph{not} an arrival time for $\dhat{e}$, then $\Steal(\dhat{e},\tau,\Delta t) = \Steal(\dhat{e},\tau,\Delta t{+}1)$.
Therefore, we may safely disregard all $\Delta t<d(\tau)$ that are not arrival times for $\dhat{e}$, and compute the $\Steal$ either for the least $\Delta t'>\Delta t$ which is an arrival time, or for $\Delta t'=d(\tau)$. By incorporating this condition, we obtain a set of all \emph{eligible} $\Steal(\dhat{e},\tau,\Delta t)$.

\paragraph{Computing the steals}

For all eligible $\dhat{e}_0$ and $\tau$, we merge the computation of
$\Steal(\dhat{e}_0,\tau,\Delta t)$ for all eligible $\Delta t$ into one as described in Algorithm~\ref{alg-search}:
Let $\dhat{e}_0=((v_a,m_a),(v_b,m_b))$. If $v_b=\tau$, the answer is trivial:
the Defender either surely catches the attack (if $\tm_1(e_0) \leq d(\tau)$) or surely fails to catch it 
(otherwise).
Otherwise, the algorithm preforms a forward search through the patrolling graph. The search is guided by a min-heap $\calH$ of items $(v,m,t,p)$,
sorted by $t$, where each item corresponds to a certain set of paths from $(v_b,m_b)$ to $(v,m)$,
all of which have the same length (total traversal time)~$t$ and whose total probability is $p$. The first heap item is $(v_b,m_b,\tm_1(e_0),1)$,
corresponding to the Defender being  in the augmented vertex $(v_b,m_b)$ at time $\tm_1(e_0)$ with probability~$1$.

Now, we explain lines~\ref{line-steal-start}--\ref{line-steal-end}.
There, we compute $\Steal(\dhat{e}_0,\tau,\Delta t)$ for $\Delta t=\ltime$.
Note that the contents of $\calH$ fully describe the possible locations of the Defender at time~$\ltime$:
each item $h\in\calH$ corresponds to the Defender being on an edge leading to $(h.v,h.m)$ with probability $h.p$, arriving there at time $h.t$ (if $h.t=\ltime$, then the Defender is already in $(h.v,h.m)$), while
failing to have caught an ongoing attack at $\tau$ yet. (Then, $1-\sum_{h\in\calH} h.p$ is the probability
that the attack has already been caught.) Thus, it suffices, for each item $h\in\calH$,
to compute the probability $p_{catch}(h)$ of visiting $\tau$ from $(h.v,h.m)$ in $G_2$ within
$d(\tau) - h.t$ time units. Then, $\Steal(\dhat{e},\tau,\ell) = \alpha(\tau)\cdot(\sum_{h\in\calH}h.p\cdot(1-p_{catch}(h)))$.

\newcommand\pe{\;\texttt{+=}\;}
\newcommand\te{\;\texttt{*=}\;}
\newcommand\ee{\;\texttt{=}\;}
\SetAlFnt{\small}
\begin{algorithm}[h!]
	\SetAlgoLined
	\DontPrintSemicolon
	\SetKwInOut{Parameter}{parameter}\SetKwInOut{Input}{input}\SetKwInOut{Output}{output}
	\SetKwData{C}{C}\SetKwData{D}{D}\SetKwData{MX}{M}\SetKwData{f}{f}
	\SetKw{Or}{or}
	\SetKw{And}{and}
	\SetKw{Not}{not}
	\SetKwData{Array}{array of}
	\SetKwData{Rat}{Rat}
	\SetKwData{Der}{Der}
	\SetKwData{MinHeap}{min-heap of Paths}
	\SetKwData{Index}{indexed by}
	\SetKwProg{Macro}{macro}{:}{}
	\Input{Patrolling graphs $G_1,G_2$, regular strategies $\sigma_1,\sigma_2$, $\dhat{e}_0 \ee ((v_a,m_a),(v_b,m_b))\in\dhat{E}$, $\tau\in T$}
	\Output{Maximum of $\Steal(\dhat{e}_0,\tau,\Delta t)$ over all $\Delta t$}
	\BlankLine
\makebox[1em][l]{$\calV$} : \texttt{array} \mbox{ indexed by eligible pairs $\dhat V$} \;
	\makebox[1em][l]{$\calH$} : \texttt{min-heap} \mbox{ of tuples $(v,m,t,p)$ sorted by $t$}    \;
	\BlankLine
	\If{$v_b=\tau$}{\Return $\tm_1(e_0) \leq d(\tau)\ ?\ 0 : \alpha(\tau)$\;}
	$steal\ee 0$ \;
	$\calH.\mathit{insert}(v_b,m_b,\tm_1(e_0),1)$\;
\While{\Not $\calH.\mathit{empty}$}{\label{line-loop-H-start}
		
		$\ltime \ee \calH.\mathit{peek}.t$\;\label{line-set-l}
$prob = 0$ \;\label{line-steal-start}
		\ForEach{$h\in\calH$}{
			$prob \pe h.p*(1-Query\_p_{catch}(h, G_2,\sigma_2,\tau))$ \;\label{line-call-query}
		}
		$steal \ee max(steal, \alpha(\tau)*prob)$ \;\label{line-steal-end}
\Repeat{$\calH.\mathit{empty}$ \Or $\calH.\mathit{peek}.t > \ltime$}{\label{line-AB-start}
			$(v,m,t,p) \ee \calH.\mathit{pop}$ \;
			$\calV(v,m) \pe p$ \;
		}\label{line-AB-end}
\ForEach{$(v,m)$ such that $\calV(v,m) > 0$\label{line-vm-start}}{
			\ForEach{$\dhat{e} \ee ((v,m),(v',m'))\in\dhat E$}{
				$t \ee \ltime + \tm_1(e)$\;
				\uIf{$t \leq d(\tau)$ \And $v'\not=\tau$}{
					$\calH.\mathit{insert}(v',m',t,\calV(v,m)*\sigma_1(\dhat{e}))$\;\label{line-H-insert}
				}
			}
			$\calV(v,m) \ee 0$\;
		} \label{line-vm-end}		
	}
	\Return $steal$
	\caption{Computes max $\Steal(\dhat{e}_0,\tau,\Delta t)$ over all $\Delta t$ for a given $\dhat{e}_0$ and $\tau$}
	\label{alg-search}
\end{algorithm}

 \paragraph{Answering the queries for $p_{catch}(h)$}

At line~\ref{line-call-query}, Algorithm~\ref{alg-search} needs to know the probability $p_{catch}(h)$. Presumably, $p_{catch}(h)$ could be computed simply by performing another similar search from $(h.v,h.m)$ in $G_2$ (omitting lines~\ref{line-steal-start}--\ref{line-steal-end}).
However, this is rather slow. Instead, we initiate a backward search
from $\tau$ in $G_2$, and then we make further enhancements in order to
answer the queries for $p_{catch}(h)$ efficiently.
The details are presented in Appendix~\ref{app-holes}.

\subsection{Preventing security holes}
Our approach to preventing large security holes is based on taking the functionality of $\sigma_1$ into account when computing the strategy $\sigma_2$ for a given $G_1 {\mapsto} G_2$. This is achieved by \emph{adapting} the strategy $\sigma_1$ to $G_2$.
Since $\sigma_1$ may not be directly executable in $G_2$ (for example, some edges of $G_1$ used by $\sigma_1$ may disappear in $G_2$), we first perform some adjustments to $\sigma_1$. Then, we \emph{improve} this initial strategy in $G_2$ by an \emph{efficient strategy improvement algorithm} described below, and thus obtain $\sigma_2$. Intuitively, since $\sigma_2$ tends to be ``similar'' to $\sigma_1$, the chance of producing \emph{unnecessary} security holes decreases. This intuition is confirmed in Experiments.

The starting point for designing our strategy improvement algorithm is \Regstar, currently the best strategy synthesis algorithm for fixed patrolling graphs recently presented by \citet{KKMR:Regstar-UAI}.
\Regstar\ repeatedly picks a random initial strategy and tries to improve its value.
The algorithm consists of two subroutines: \emph{Evaluation}, \ie, computing of the value and a gradient of a given strategy and \emph{Optimization} using gradient descent. 
After hundreds of trials, the best strategy found is chosen. However, the percentage of trials converging to the best strategy found can be rather low ($\approx 2\%$) \citep[see][Sec.~3.5]{KKMR:Regstar-UAI}.
For this reason, our initial attempt to construct $\sigma_2$ by applying the strategy-improvement subroutine of \Regstar to $\sigma_1$ in the graph $G_2$ \emph{failed}.
This calls for \Regstar re-design.

First, we replace the optimization scheme using dedicated tools for differentiable programming (PyTorch with Adam optimizer).
We also add decaying Gaussian noise to the gradient allowing for different outcomes when optimizing from $\sigma_1$ and hence enlarging the chance of hitting high-valued $\sigma_2$. In contrast, the optimization loop of \Regstar\ is purely deterministic. Furthermore, the \Regstar's evaluation computes gradients in forward mode. We re-design this part by employing the reverse mode, yielding improvement by a factor of $|\dhat E|$.

Our modifications drastically improve \Regstar's convergence ratio and speed (see the analysis in Experiments) and allow for \emph{on the fly strategy adaptation}.
Implementation details are described and the code is provided in
\href{https://gitlab.fi.muni.cz/formela/2022-UAI-changing-env}{https://gitlab.fi.muni.cz/formela/2022-UAI-changing-env}.

\subsection{Mitigating security holes}

In general, the structural  difference between $G_1$ and $G_2$ in a changing environment $G_1 {\mapsto} G_2$ can make the creation of security holes \emph{unavoidable}, as demonstrated by the following example. 

\begin{example}
\label{exa-unavoidable}
Consider the setup of Example~\ref{exa-change}. There is only \emph{one} $\sigma_2$ such that $\Val_{G_2}(\sigma_2) = 100$ (the ``anticlockwise walk''), and hence there is no reasonable alternative to $\sigma_2$. Since $\Val_{G_1}(\sigma_1) = \Val_{G_2}(\sigma_2) = 100$, we \emph{inevitably} obtain the largest conceivable security hole equal to~$100$.
\end{example}

However, we show that under the conditions given below, the security holes can be \emph{mitigated without harming the protection achieved by~$\sigma_2$} by \emph{randomized strategy switching}. 
Let us assume the following:
\begin{itemize}
\item For every $(v,m)$ visited by $\sigma_1$ with positive probability, there is $(v,m')$ such that $\Val_{G_2}(\sigma_2)(v,m') =  \Val_{G_2}(\sigma_2)$.
\item $\sigma_1$ is executable in $G_2$ and\footnote{This is \emph{not} a typo. If $G_2$ has the same topology as $G_1$ but edge traversal times change, then setting $\sigma_2 = \sigma_1$ \emph{may} cause a security hole. The assumption  $\hole_{G_1 \mapsto G_2}(\sigma_1,\sigma_1)=0$ says that this does not happen.} $\hole_{G_1 \mapsto G_2}(\sigma_1,\sigma_1)=0$.
\end{itemize}
Under these conditions, the Defender may perform a \emph{randomized switch from $\sigma_1$ to $\sigma_2$}. That is, the Defender flips a $\kappa$-biased coin when it arrives in a vertex and switches to $\sigma_2$ only with probability $\kappa$. With the remaining probability $1-\kappa$, the Defender continues executing $\sigma_1$ and flipping the coin in the next vertex again. This goes on until the switch to $\sigma_2$ is performed. We have the following:

\begin{restatable}[]{theorem}{thmhole}
\label{thm-hole}
The expected number of time units needed to perform the $\kappa$-randomized switch is bounded by $\textit{max-time$_2$}/\kappa$, where \textit{max-time$_2$} is the maximal traversal time of an edge in $G_2$. 

The security hole caused by the switch is bounded by
\[
   \varrho + \left(1 - (1-\kappa)^{\dm}\right) \cdot \alpha_{\max}(G_2)
\]
where $\alpha_{\max}(G_2)$ is the maximal target cost in $G_2$ and $\varrho$ is defined as
\[
  \max\big\{0, \ \AVal_{G_2}(\sigma_1) {-} \max\{\AVal_{G_1}(\sigma_1), \AVal_{G_2}(\sigma_2)\} \big\}
\]
Hence, the security hole can be pushed arbitrarily close to~$\varrho$ by choosing a suitably small~$\kappa >0$. 
\end{restatable}

\noindent
A proof of Theorem~\ref{thm-hole} is in Appendix~\ref{app-holes}. The usefulness of randomized strategy switching is documented on a concrete scenario in Section~\ref{sec-experiments}.

\section{Experiments}
\label{sec-experiments}

\subsection{Strategy improvement analysis} \label{algo-analysis}

We assess our strategy synthesis algorithm in comparison with \Regstar
on the set of patrolling graphs used to evaluate \Regstar by \citet{KKMR:Regstar-UAI}. These graphs model office buildings, and their structure is recalled in Appendix~\ref{app-experiments}.

Here we present the outcomes for a graph modeling a \mbox{2-floor} building achieved for $\Mem$ with $1,\ldots,8$ elements, cf.{} \citep[Experiment~5.3]{KKMR:Regstar-UAI}. Fig.~\ref{fig:exp1_2floor_building} shows boxplot statistics of values of strategies found by 200 trials of \Regstar (blue) and our improved (red) method. Note that our method consistently produces values concentrated around the best value found, i.e., the chance of producing a strategy with a high value from a random initial strategy is \emph{high}. Furthermore, the value of the \emph{best} strategy found by our method is \emph{higher} than the one found by \Regstar{} in all cases except for $|\Mem| = 8$ (where the difference is negligible). Similar results are obtained for \emph{all} patrolling graphs analyzed by \citet{KKMR:Regstar-UAI}. These datasets together with a detailed setup description are in Appendix~\ref{app-experiments}.

\begin{figure}[t]
	\centering
	\includegraphics[width=\columnwidth]{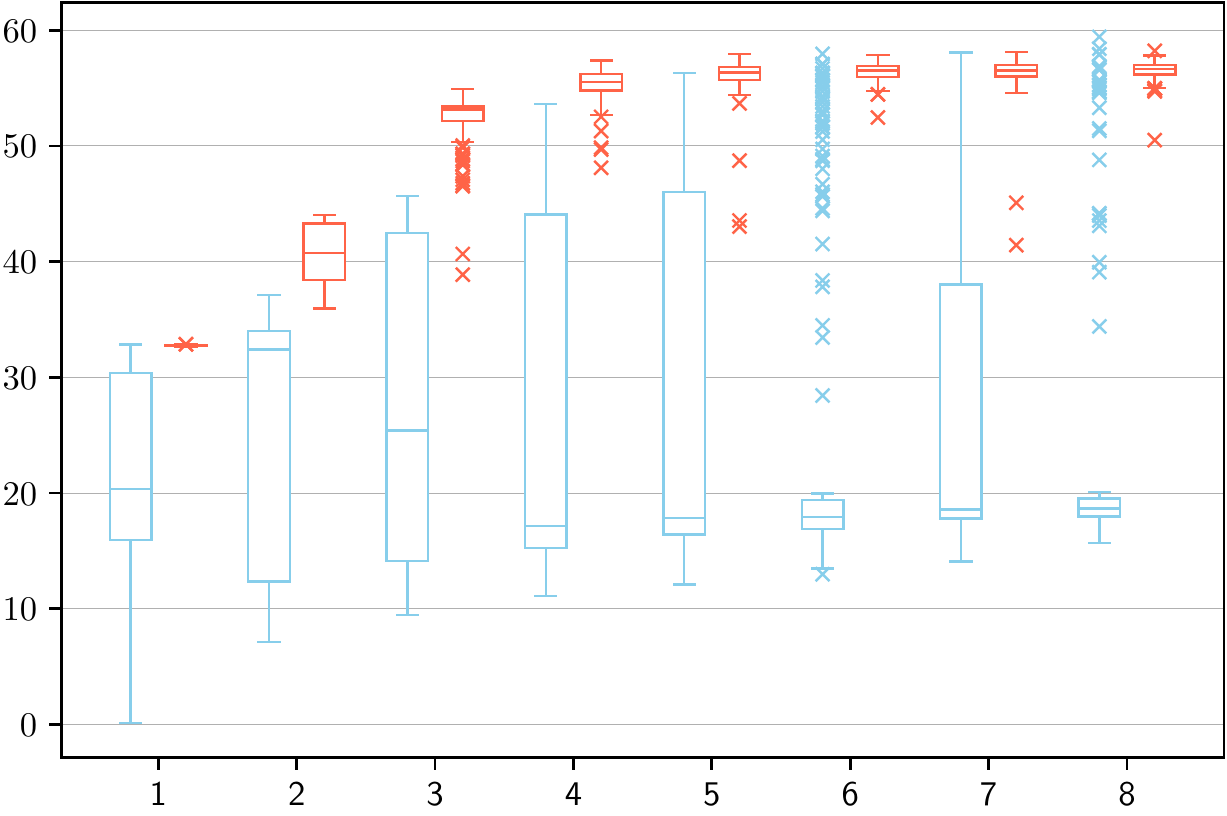}
	\caption{Values of strategies synthesized by \Regstar (blue) and our method (red) for a 2-floor building graph where $\Mem$ has $1,\ldots,8$ elements. 
	The red values are tightly distributed close to the maxima.}
	\label{fig:exp1_2floor_building}
\end{figure}

Next, we report runtimes of the forward (value) and backward (gradient) computations of the strategy-evaluation module.
Tab.~\ref{tab:building-two-floors-times} summarizes the mean of 200 passes through the strategy evaluation on the same 2-floor building graph with various $\Mem$ sizes.
Reverse-mode gradient computation improved the backward times by \emph{three orders of magnitude} (note that the time is given in \emph{seconds} for \Regstar and in \emph{miliseconds} for our method). 

\begin{table}[t]
	\small
	\centering
\begin{tabular*}{\linewidth}{@{}l@{\extracolsep{\fill}}d{3,1}d{3,2}d{4,3}d{2,1}}
		\toprule
		& \multicolumn{2}{c}{ forward  [ms]} & \multicolumn{2}{c}{ backward } 
		\\
		$m$ & \multicolumn{1}{c}{\Regstar} & \multicolumn{1}{c}{Ours} & \multicolumn{1}{c}{\Regstar[s]} & \multicolumn{1}{c}{Ours [ms]}
		\\ \midrule
2	& 50 + 3  & 48 + 5	    & 0.64 + 0.04   & 4 + 0	   \\
4	& 217 + 4 & 198 + 19	& 10.2 + 0.7    & 14 + 2		\\
6	& 492 + 4 & 451 + 44	& 57.4 + 3.9    & 33 + 3		\\
8	& 913 + 9 & 805 + 80	& 186.3 + 12.0  & 60 + 7		\\
\bottomrule
	\end{tabular*}
\caption{Effect of differentiation of $\sigma\mapsto\Val(\sigma)$ in reverse mode (Ours) compared to forward mode (\Regstar).}
	\label{tab:building-two-floors-times}
\end{table}

\subsection{Changing environment}

We evaluate our algorithms for concrete changing environments.
Specifically, we quantify the impact of our approach to preventing security holes and examine the effectiveness of randomized strategy switching on mitigating security holes.
 
We fix one patrolling graph $G_1$ consisting of $15$ locations in the downtown of Vancouver. The target costs are set between $80$ and $100$ at random. Furthermore,
we select $72$ edges connecting the targets with lengths measured in taxicab distance in hundreds of meters. Attack times are fixed to $64$, giving the Defender chance to discover an attack starting $6.4$km far away. For~$G_1$, we find and fix a strategy $\sigma_{1}$ with $\Val_{G_1}(\sigma_1)=42.1$.

We perform three sets of experiments, modifying $G_1$ to $G_2$ by either changing the target costs, edge lengths, or removing some edges.
The experiments are parameterized by the \emph{change size}, denoted by $\changesize$.
For all types we report two $\changesize$ values representing small and large change impact.
More values are reported in Appendix~\ref{app-experiments}.

\emph{Utility changes}\quad
The cost of each node is increased by its $\changesize\%$ with probability 1/3, decreased by $\changesize\%$ with probability 1/3, or left unchanged.
Note that utility changes can modify $\alpha_{\max}$ and thus influence $\Val$.
To compare, we normalize all results by $100/\alpha_{\max}$ for each $G_2$ and its~$\alpha_{\max}$.

\emph{Variable edge length}\quad
As in the previous case, the length of each edge is increased/decreased by $\changesize\%$ or kept unchanged (with the same probability).

\emph{Removed edges}\quad
We randomly delete $\changesize$ edges so that $G_2$ remains strongly connected.

For each $\changesize$, we generate $10$ modified graphs~$G_2$.
For every $G_2$, we take the highest value and security gap from $10$ optimization trials with $0$, $50$, $100$, $200$, and $400$ optimization steps initiated in $\sigma_{1}$ (recall that the optimization step in our algorithm uses noising and hence the output is different for each of the $10$ trials).
We report the means and standard deviations over all $10$ modified graphs $G_2$.
The same statistics are reported for the runs that start from random initialization instead of from~$\sigma_1$.
This is repeated for every $\changesize$. Hence, for each line of Tab.~\ref{tab:experiments}, we run $2\times 100$ optimization trials.
For setup details, see Appendix~\ref{app-experiments}.

\let\hc\relax
\definecolor{skyblue}{HTML}{97CCE8}
\definecolor{ourorange}{HTML}{ED6D52}
\definecolor{ourgreen}{HTML}{73FA79}
\newcommand\hcglob[3]{\cellcolor{skyblue!\fpeval{((#1-#2)/(#3-#2))**2*100}!ourorange}#1}
\newcommand\hcglobg[3]{\cellcolor{ourorange!\fpeval{((#1-#2)/(#3-#2))*100}!ourgreen}#1}
\newcommand\tabHead{\multirow{2}{*}{$\changesize$}& \multirow{2}{*}{steps}& \multicolumn{2}{c}{$\Val$}& \multicolumn{2}{c}{Security Hole} \\
	&& & \multicolumn{1}{c}{from $\sigma_{1}$}& \multicolumn{1}{c}{from rnd}& \multicolumn{1}{c}{from $\sigma_{1}$}& \multicolumn{1}{c}{from rnd} \\
}

\begin{table}[t]
	\small
	\centering
		\begin{adjustbox}{max width=\linewidth}
		\begin{tabular}{@{}lr@{\ }@{\extracolsep{\fill}}rd{3,3}d{3,2}d{3,2}d{3,2}}
		\toprule
			\multirow{2}{*}{}&\tabHead
		\midrule
			\multirow{10}{*}{\rotatebox{90}{utility changes}}
			& \multirow{5}{*}{5}
			\gdef\hc#1{\hcglob{#1}{12.7}{43.8}}
			\gdef\hcg#1{\hcglobg{#1}{0}{33.7}}
			 & 0		& \hc{40.9}	+ 0.9	& \hc{12.7}	+ 3.2	& \hcg{0.0}	+ 0.0	& \hcg{ 4.6}	+ 2.7	\\
			&& 50	& \hc{43.4}	+ 0.6	& \hc{27.4}	+ 0.8	& \hcg{2.8}	+ 1.2	& \hcg{14.4}	+ 1.2	\\
			&& 100	& \hc{43.6}	+ 0.5	& \hc{37.3}	+ 1.2	& \hcg{3.9}	+ 1.9	& \hcg{25.3}	+ 2.6	\\
			&& 200	& \hc{43.8}	+ 0.6	& \hc{41.3}	+ 0.6	& \hcg{5.2}	+ 2.8	& \hcg{30.2}	+ 4.4	\\
			&& 400	& \hc{43.8}	+ 0.6	& \hc{42.7}	+ 0.5	& \hcg{6.9}	+ 4.4	& \hcg{33.7}	+ 3.9	\\
			\cmidrule{2-7}
			& \multirow{5}{*}{30}
			\gdef\hc#1{\hcglob{#1}{14.7}{54}}
			\gdef\hcg#1{\hcglobg{#1}{0}{40}}
			& 0		& \hc{40.9}	+ 0.9	& \hc{14.7}	+ 4.2	& \hcg{ 0.0}	+ 0.0	& \hcg{ 6.1}	+ 2.2	\\
			&& 50	& \hc{50.4}	+ 3.4	& \hc{38.2}	+ 4.4	& \hcg{ 9.7}	+ 1.9	& \hcg{19.6}	+ 3.3	\\
			&& 100	& \hc{51.4}	+ 3.6	& \hc{48.7}	+ 2.8	& \hcg{11.3}	+ 3.1	& \hcg{30.1}	+ 5.6	\\
			&& 200	& \hc{52.4}	+ 3.8	& \hc{52.8}	+ 3.5	& \hcg{14.0}	+ 2.3	& \hcg{39.4}	+ 4.1	\\
			&& 400	& \hc{52.9}	+ 3.7	& \hc{54.0}	+ 3.3	& \hcg{15.2}	+ 3.2	& \hcg{40.0}	+ 5.6	\\
			\midrule
			\multirow{10}{*}{\rotatebox{90}{variable edge length}}
			& \multirow{5}{*}{5}
			\gdef\hc#1{\hcglob{#1}{9.4}{41}}
			\gdef\hcg#1{\hcglobg{#1}{0.6}{33.8}}
			& 0     & \hc{36.4} + 2.2    & \hc{9.4} + 0.3     & \hcg{0.6} + 1.0     & \hcg{2.9} + 0.3     \\
			&& 50    & \hc{40.3} + 0.9    & \hc{24.3} + 0.6    & \hcg{5.0} + 1.9     & \hcg{13.0} + 1.1    \\
			&& 100   & \hc{40.6} + 0.9    & \hc{34.7} + 0.7    & \hcg{6.5} + 2.6     & \hcg{26.2} + 2.5    \\
			&& 200   & \hc{40.9} + 0.9    & \hc{39.3} + 0.7    & \hcg{8.8} + 3.1     & \hcg{31.9} + 2.3    \\
			&& 400   & \hc{41.0} + 0.9    & \hc{40.4} + 0.8    & \hcg{9.4} + 3.1     & \hcg{33.8} + 2.7    \\
			\cmidrule{2-7}
			& \multirow{5}{*}{30}          
			\gdef\hc#1{\hcglob{#1}{8.1}{44.6}}
			\gdef\hcg#1{\hcglobg{#1}{0.1}{34.1}}
			& 0     & \hc{16.3} + 10.1   & \hc{8.1} + 1.1     & \hcg{0.1} + 0.4     & \hcg{1.7} + 1.2     \\
			&& 50    & \hc{33.6} + 5.0    & \hc{24.0} + 2.5    & \hcg{5.4} + 5.4     & \hcg{13.2} + 2.8    \\
			&& 100   & \hc{36.4} + 3.7    & \hc{37.8} + 2.4    & \hcg{10.2} + 4.6    & \hcg{27.5} + 2.3    \\
			&& 200   & \hc{38.7} + 3.2    & \hc{42.7} + 2.2    & \hcg{14.2} + 5.3    & \hcg{33.9} + 2.6    \\
			&& 400   & \hc{40.8} + 2.6    & \hc{44.4} + 2.9    & \hcg{20.9} + 3.8    & \hcg{34.1} + 3.5    \\
			\midrule
			\multirow{10}{*}{\rotatebox{90}{removed edges}}
			& \multirow{5}{*}{1}
			\gdef\hc#1{\hcglob{#1}{9.6}{42.1}}
			\gdef\hcg#1{\hcglobg{#1}{0}{34.5}}
			& 0     & \hc{39.5} + 5.3    & \hc{9.6} + 0.3     & \hcg{0.0} + 0.0     & \hcg{3.1} + 0.4     \\
			&& 50    & \hc{42.0} + 0.2    & \hc{24.6} + 0.3    & \hcg{0.9} + 1.5     & \hcg{13.3} + 0.8    \\
			&& 100   & \hc{42.1} + 0.1    & \hc{35.1} + 0.7    & \hcg{0.9} + 1.5     & \hcg{24.9} + 2.1    \\
			&& 200   & \hc{42.1} + 0.1    & \hc{39.4} + 0.3    & \hcg{0.9} + 1.4     & \hcg{31.9} + 3.3    \\
			&& 400   & \hc{42.1} + 0.1    & \hc{40.9} + 0.5    & \hcg{1.3} + 1.2     & \hcg{34.5} + 2.3    \\
			\cmidrule{2-7}
			& \multirow{5}{*}{8}
			\gdef\hc#1{\hcglob{#1}{9.8}{40.9}}
			\gdef\hcg#1{\hcglobg{#1}{0}{34.4}}
			& 0     & \hc{17.0} + 11.4   & \hc{9.8} + 0.9     & \hcg{0.0} + 0.0     & \hcg{3.3} + 0.6     \\
			&& 50    & \hc{36.1} + 9.1    & \hc{24.8} + 0.3    & \hcg{8.2} + 5.3     & \hcg{13.2} + 1.7    \\
			&& 100   & \hc{39.7} + 1.9    & \hc{35.3} + 0.9    & \hcg{10.3} + 5.1    & \hcg{25.5} + 2.4    \\
			&& 200   & \hc{40.4} + 1.4    & \hc{39.0} + 0.7    & \hcg{10.9} + 5.1    & \hcg{29.8} + 2.1    \\
			&& 400   & \hc{40.9} + 1.0    & \hc{40.5} + 0.7    & \hcg{11.8} + 5.5    & \hcg{34.4} + 2.6    \\
			\bottomrule
		\end{tabular}
		\end{adjustbox}
	\caption{
	Values and security holes of strategies in a changed graph $G_2$ optimized from the old strategy or from scratch.
	Initialization in $\sigma_1$ leads to higher $\Val$ and much smaller security holes in fewer iterations, unless the changes ($\changesize$) are too large.
	}
	\label{tab:experiments}
\end{table}

\paragraph{Summary}
All experiments unanimously confirm that, for small $\changesize$, the initialization in $\sigma_1$ leads to higher $\Val$ and much smaller security holes in fewer iterations.
For small $\changesize$, strategies obtained after $50$ optimization steps from $\sigma_1$ are not outperformed even by $400$
steps of optimization initiated in a randomly chosen strategy.

Only for large $\changesize$ in edge length, the optimizations initiated in a random strategy reach higher values than the optimization initiated in $\sigma_1$.
However, the security gap is huge.

In all our experiments, the average time needed for performing one optimization step is $110$ milliseconds, which is sufficient for performing our algorithm on the fly.

\paragraph{Mitigating Security Holes}

The conditions enabling randomized strategy switching are satisfied for all $\sigma_2$ summarized in utility changes of Tab.~\ref{tab:experiments}.
We have that $\Val_{G_1}(\sigma_1) = 42.1$ and $\Val_{G_2}(\sigma_1) = 40.9$, which means $\AVal_{G_1}(\sigma_1) = 57.9$ and $\AVal_{G_2}(\sigma_1) = 59.1$.
By Theorem~\ref{thm-hole}, the security hole can be reduced arbitrarily close to $1.2$ for \emph{all} $\sigma_2$, including those constructed from randomly chosen initial strategies.
Note that the improvement is \emph{significant} in almost all cases.

\section{Conclusions}

Our experiments show that our strategy adaptation algorithm is sufficiently efficient to be run on-the-fly, outperforming the best existing strategy synthesis algorithm \Regstar\ by three orders of magnitude. Furthermore, the experiments demonstrate the effectiveness of the designed methods for preventing and mitigating security holes.

An interesting open question is whether the Defender can effectively decrease the potential negative impact of environmental changes by \emph{preventive} adaptations of its current strategy.
This approach is applicable in cases when the probability of these changes happening in a near future is known.

\begin{acknowledgements}

Research was sponsored by the Army Research Office and was accomplished under
Grant Number W911NF-21-1-0189.
V\'{\i}t Musil was supported by Operational Programme Research, Development and
Education -- Project Postdoc2MUNI (No.\ CZ.02.2.69/0.0/0.0/18\_053/0016952).

\textit{Disclaimer.}\quad The views and conclusions contained in this document are those of the authors
and should not be interpreted as representing the official policies, either
expressed or implied, of the Army Research Office or the U.S.\ Government. The
U.S.\ Government is authorized to reproduce and distribute reprints for
Government purposes notwithstanding any copyright notation herein.
\end{acknowledgements}

\appendix
\newpage

\section{Changing Environment}
\label{app-changing}

\paragraph{Evaluating a switching strategy}

First, we demonstrate a situation when the Attacker may increase its expected utility by initiating an attack in the middle of a Defender's move.

\begin{example}
	\label{exa-delta}
	Let $G_1$ be the patrolling graph of Fig.~\ref{fig-delta}. Assume that at time $0$, the Defender starts walking from $v_1$ to $v_2$. Let $p$ denote the probability that when the Defender reaches $v_2$ at time $4$, he starts going back to $v_1$. Moreover, assume that at time $10$, the patrolling environment changes in such a way that $\alpha_2(v_1)>\alpha_1(v_1)$. If the Attacker attacks $v_1$ at time $0$ (corresponding to $\delta=0$), then his expected utility is $p\cdot \alpha_1(v_1)$. However, if he instead attacks at time $2$ (corresponding to $\delta=2$),
	then his expected utility is $p\cdot \alpha_2(v_1)$, because the attack is completed after the environmental change. Finally, note that if the Attacker postpones his attack to 
	the moment when the Defender leaves $v_2$ (\ie, at time $4$, corresponding to $\delta=0$), then the Defender is given one more opportunity to discover the attack (a visit to $v_1$ at time $12$ still catches this attack).
	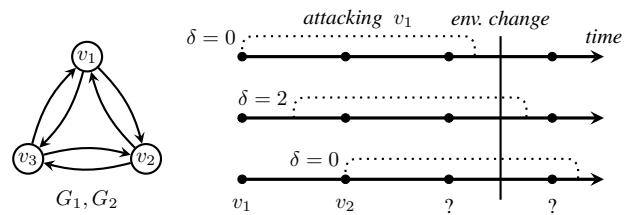
\begin{figure}[hbt]
		\begin{center}
\begin{tikzpicture}[x=1.7cm,y=1.7cm,scale =.8,every node/.style={scale=.8}]
			\node (v1) at (-.5,0)      [ran]  {$v_1$}; 
			\node (v2) at (-1.07,-1)  [ran]  {$v_3$}; 
			\node (v3) at (0.07,-1)  [ran]  {$v_2$}; 
			\path[tran] (v1)  edge   [bend left=15] (v2);
			\path[tran] (v2)  edge   [bend left=15] (v1);
			\path[tran] (v2)  edge   [bend left=15] (v3);
			\path[tran] (v3)  edge   [bend left=15] (v2);
			\path[tran] (v1)  edge   [bend left=15] (v3);
			\path[tran] (v3)  edge   [bend left=15] (v1);
			\node[draw=none] at (-.5,-1.4) {$G_1,G_2$};
			
			\draw[tran, very thick,->] (1,0) -- (4.5,0);
			\node[draw=none] at (4.5,.2) {\textit{time}};
			\filldraw[black] (1,0) circle (2pt) node[below=1ex]{};
			\filldraw[black] (2,0) circle (2pt) node[below=1ex]{};
			\filldraw[black] (3,0) circle (2pt) node[below=1ex]{};
			\filldraw[black] (4,0) circle (2pt) node[below=1ex]{};
			\draw[thick,dotted,rounded corners] (1,0) -- +(0,.2) -- node[above] {\textit{attacking } $v_1$}  +(2.25,.2) -- +(2.25,0);
			\node[draw=none] at (.7,.2) {$\delta=0$};
			
			\draw[tran, very thick,->] (1,-.6) -- (4.5,-.6);
			\filldraw[black] (1,-.6) circle (2pt) node[below=1ex]{};
			\filldraw[black] (2,-.6) circle (2pt) node[below=1ex]{};
			\filldraw[black] (3,-.6) circle (2pt) node[below=1ex]{};
			\filldraw[black] (4,-.6) circle (2pt) node[below=1ex]{};
			\draw[thick,dotted,rounded corners] (1.5,-.6) -- +(0,.2) -- node[above] {}  +(2.25,.2) -- +(2.25,0);
			\node[draw=none] at (1.2,-.4) {$\delta=2$};
			
			\draw[tran, very thick,->] (1,-1.2) -- (4.5,-1.2);
			\filldraw[black] (1,-1.2) circle (2pt) node[below=1ex]{$v_1$};
			\filldraw[black] (2,-1.2) circle (2pt) node[below=1ex]{$v_2$};
			\filldraw[black] (3,-1.2) circle (2pt) node[below=1ex]{?};
			\filldraw[black] (4,-1.2) circle (2pt) node[below=1ex]{?};
			\draw[thick,dotted,rounded corners] (2,-1.2) -- +(0,.2) -- node[above] {}  +(2.25,.2) -- +(2.25,0);
			\node[draw=none] at (1.7,-1) {$\delta=0$};
			
			\draw[thick] (3.5,.2) -- (3.5,-1.4);
			\node[draw=none] at (3.5,.36) {\textit{env. change}};
			\end{tikzpicture}
		\end{center}
		\caption{Exemplifying the advantage of $\delta>0$. Traversing every edge takes $4$ time units, completing an attack takes $9$ time units.}
		\label{fig-delta}
	\end{figure}
\end{example}

In what follows, we present exact definitions of all concepts which must be redefined for the changing environment.
Let $G_1=(V,T,E_1,\tm_1,d,\alpha_1)$ and $G_2=(V,T,E_2,\tm_2,d,\alpha_2)$ be two patrolling graphs.

\paragraph{Attacker's strategy}

An \emph{observation} in $G_1 \mapsto_t G_2$ is a pair $(o,\delta)$, where $o = v_1,\ldots, v_n,v_n {\rightarrow} v_{n+1}$
is a path in~$G$ and $\delta$ is an integer satisfying $0 \leq \delta < \tm_i(v_n {\rightarrow} v_{n+1})$,
where $i=1$ if the move $v_n {\rightarrow} v_{n+1}$ is initiated before time~$t$ (\ie, $\sum_{j=1}^{n-1}\tm_1(v_{j} {\rightarrow} v_{j+1})<t$), and $i=2$ otherwise.
The set of all observations is denoted by $\Obs$. 

An \emph{Attacker's strategy} is a function
\[
 \pi\colon \Obs \rightarrow \{\wait,\attack_\tau:\tau\in T\}.
\]
We require that
if $\pi((v_1,\ldots, v_n,v_n {\rightarrow} u),\delta) = \attack_\tau$ for some $\tau \in T$,
then $\pi((v_1,\ldots, v_n,v_n {\rightarrow} u),\delta') = \wait$ for all $0\leq\delta'<\delta$
and $\pi((v_1,\ldots, v_i,v_i {\rightarrow} v_{i+1}),\delta') = \wait$ for all $1\leq
i<n$ and any $\delta'$, ensuring that the Attacker can attack at most once.

\paragraph{Evaluating Defender's strategy}

Let $\sigma_1$ and $\sigma_2$ be Defender's strategies in $G_1$ and $G_2$, respectively, and let $\pi$ be an Attacker's strategy.

Let us fix a switching time $t\in\Nset$ and
an initial augmented vertex $\hat{v}$ where the Defender starts patrolling.
The \emph{expected Attacker's utility} $\EU^{\sigma_1 \mapsto_t \sigma_2,\pi}(\hat{v})$ is defined as
\begin{equation*}
	\sum_{\hat{e},\tau,\delta,t_0}  \mathbf{P}^{\sigma_1 \mapsto_t \sigma_2,\pi}(\hat{e},\tau,\delta,t_0) \cdot \Steal^{\sigma_1 \mapsto_t \sigma_2}(\hat{e},\tau,\delta,t_0)\,.
\end{equation*}
where $\mathbf{P}^{\sigma_1 \mapsto_t \sigma_2,\pi}(\hat{e},\tau,\delta,t_0)$ is the probability of initiating an attack at $\tau$ at time $t_0$ when the Defender has been going along~$\hat{e}$ for $\delta$ time units, and $\Steal^{\sigma_1 \mapsto_t \sigma_2}(\hat{e},\tau,\delta,t_0)$ denotes the expected cost ``stolen'' by this attack.

More precisely, let $\Attack(\pi,\hat{e},\tau,\delta,t_0)$ be the set of all $(\hat{v}_1,\ldots,\hat{v}_{n+1})$ such that $\pi((v_1,\ldots,v_n,v_n {\to} v_{n+1}),\delta) = \tau$, $\hat{e} = \hat{v}_n \to \hat{v}_{n+1}$ and $t_0 = \delta+\sum_{j=1}^{n-1}\tm_{i_j}(v_{j} {\rightarrow} v_{j+1})$ where again each $i_j\in\{1,2\}$ denotes in which graph the edge was traversed (\ie,
$i_j=1$ iff $\sum_{k=1}^{j-1}\tm_1(v_{k} {\rightarrow} v_{k+1})<t$). We put
\begin{equation*}
	\mathbf{P}^{\sigma_1 \mapsto_t \sigma_2,\pi}(\hat{e},\tau,\delta,t_0) = \sum_{h \in \Attack(\pi,\hat{e},\tau,\delta,t_0)} \Prob(h)
\end{equation*}
Furthermore, let $\mathbf{M}^{\sigma_1 \mapsto_t \sigma_2}(\hat{e},\tau,\delta,t_0)$ be the probability of missing (i.e., not visiting) an augmented vertex of the form $\hat{\tau}$ in the time interval $[t_0+1,t_0+d(\tau)]$ provided the Defender starts going along $\hat{e}$ at time $t_0-\delta$. We define 
\begin{equation*}
\Steal^{\sigma_1 \mapsto_t \sigma_2}(\hat{e},\tau,\delta,t_0) = \alpha_i(\tau) \cdot \mathbf{M}^{\sigma_1 \mapsto_t \sigma_2}(\hat{e},\tau,\delta,t_0)
\end{equation*}
where $i=1$ if $t_0+d(\tau)<t$ (\ie, the attack is \emph{completed} before the environment changes) and $i=2$ otherwise.

Finally, the \emph{Attacker's value of $\sigma_1 \mapsto \sigma_2$ in $G_1 \mapsto G_2$} is defined as
\begin{equation*}
    \AVal_{G_1 \mapsto G_2}(\sigma_1 {\mapsto} \sigma_2) =  \min_{\hat{v}}\ \sup_{\pi}\ \sup_{t}\
    \EU^{\sigma_1 \mapsto_t \sigma_2,\pi}(\hat{v})\,.
\end{equation*}

\section{Security Holes}
\label{app-holes}

\subsection{Estimating security holes}

 \renewcommand\pe{\;\texttt{+=}\;}
 \renewcommand\te{\;\texttt{*=}\;}
 \renewcommand\ee{\;\texttt{=}\;}

In this subsection, we present Algorithm~\ref{alg-query}, which efficiently answers queries for $p_{catch}(h)$, required by Algorithm~\ref{alg-search}.

\begin{algorithm}[h!]
	\SetAlgoLined
	\DontPrintSemicolon
	\SetKwInOut{Parameter}{parameter}\SetKwInOut{Input}{input}\SetKwInOut{Output}{output}
	\SetKwData{C}{C}\SetKwData{D}{D}\SetKwData{MX}{M}\SetKwData{f}{f}
	\SetKw{Or}{or}
	\SetKw{And}{and}
	\SetKw{Not}{not}
	\SetKwData{Array}{array of}
	\SetKwData{Rat}{Rat}
	\SetKwData{Der}{Der}
	\SetKwData{MinHeap}{min-heap of Paths}
	\SetKwData{Index}{indexed by}
	\SetKwProg{Macro}{macro}{:}{}
	\Input{Heap item $h_0$, patrolling graph $G_2$, regular strategy $\sigma_2$, $\tau\in T$}
	\Output{$p_{catch}(h_0)$}
	\BlankLine
\makebox[1em][l]{$\calV$} : \texttt{array} \mbox{ indexed by eligible pairs $\dhat V$} \;
	\makebox[1em][l]{$\calH$} : \texttt{min-heap} \mbox{ of tuples $(v,m,t,p)$ sorted by $t$}    \;
	\BlankLine
	$prob\ee 0$ \;
	\ForEach{$m\in\mem$}{
		$\calH.\mathit{insert}(\tau,m,0,1)$\; }
\While{\Not $\calH.\mathit{empty}$}{
		\Repeat{$\calH.\mathit{empty}$ \Or $\calH.\mathit{peek}.t > \ltime$}{
			$(v,m,t,p) \ee \calH.\mathit{pop}$ \;
			$\calV(v,m) \pe p$ \;
			\If{$(v,m) = (h_0.v,h_0.m)$}{$prob \pe p$ \;}
		}
\ForEach{$(v,m)$ such that $\calV(v,m) > 0$}{
			\ForEach{$\dhat{e} \ee ((v',m'),(v,m))\in\dhat E$}{
				$t \ee \ltime + \tm_2(e)$\;
				\uIf{$t \leq d(\tau) - h_0.t$ \And $v'\not=\tau$\label{line-t}}{
					$\calH.\mathit{insert}(v',m',t,\calV(v,m)*\sigma_2(\dhat{e}))$\;
				}
			}
			$\calV(v,m) \ee 0$\;
		}
	}
	\Return $prob$
	\caption{Computes $p_{catch}(h_0)$ for a given heap item $h_0$}
	\label{alg-query}
\end{algorithm}

Let $S$ denote the set of all heap items that were created during the search.
Now, each heap item $(v,m,t,p)$ corresponds to a certain set of paths from $(v,m)$ to $\dhat{\tau}$;
the meaning of $t$ and $p$ is the same as in the forward search (Algorithm~\ref{alg-search}).
Note that $p_{catch}(h_0)$ is equal to the sum of $h.p$ over all $h\in S$
such that $(h.v,h.m)=(h_0.v,h_0.m)$. Moreover, note that thanks to the backward manner
of the search, the created heap items are independent of $(h_0.v,h_0.m)$,
and the value of $h_0.t$ affects only the length of the search.
Thus, if we continue the search up to $t \leq d(\tau)$ (cf. line~\ref{line-t}),
then the computation is totally independent of $h_0$, and can be done,
for each $\tau$, just once as a precomputation step.

Then, $p_{catch}(h_0)$ can be computed as the sum of $h.p$ over all $h\in S$
such that $(h.v,h.m)=(h_0.v,h_0.m)$ and $h.t \leq d(\tau) - h_0.t$.
It remains to show how to compute this value quickly.
Note that naively going through all the items would result in $\Theta(|S|)$ time answering a query.
Thus, we split the items into a $|\dhat{V}|$-indexed array of buckets of items sharing the same $(v,m)$
(thereby reducing the average answering time by a factor of $|\dhat{V}|$), and we keep the items in each bucket
sorted by $t$ (further halving the average time). Thus, exactly those items that are included in the sum are quickly found.
However, the resulting algorithm still proved too slow to perform all our experiments.
Therefore, we come up with the following trick: We precompute the prefix sums in each bucket, \ie, to each item $h$,
we add another component $h.s$ which equals the sum of $h.p$ over all items $h'$ in the same bucket with
$h'.t \leq h.t$ (we also add a sentinel item $h$ with $h.t<0$ and $h.s=0$, which corresponds to the empty sum).
Then, $p_{catch}(h_0)$ is equal to $h'.s$ for some $h'$ whose position can be found by binary search.
This reduces the average answering time to $O(\log(|S|/|\dhat{V}|))$.

\subsection{Preventing security holes}

We synthesize the strategies by a gradient ascent.
In one trial, we start from an initial strategy, repeatedly compute its value and update the strategy in the direction of the value's gradient.
In bounded time frames, the strategy \emph{evaluation} is the bottleneck and limits the total number of trials and iterations per trial.

While \Regstar has impressive time performance, \citet{KKMR:Regstar-UAI} conclude that \Regstar rarely converges to high-valued strategies.
Therefore, we redesign the optimization part completely.
The evaluation part of \Regstar is kept, but we significantly improve the efficiency of the gradient computation.

\paragraph{Strategy evaluation}

We have implemented a PyTorch module in C++ for the strategy evaluation.
The evaluation of $\Val(\sigma)$ (so-called \emph{forward pass}) is the same as in \Regstar
but the automatic differentiation of $\Val$ (so-called \emph{backward pass}) is computed differently.

\Regstar computes the gradient in \emph{forward mode},
\ie, for each node $n$ of the computation graph, $\partial n/\partial \sigma(\dhat{e})$ is computed for each $\dhat{e}\in\dhat{E}$,
going from the input nodes (corresponding to $\sigma(\dhat{e})$ for the individual augmented edges $\dhat{e}$) to the output node (corresponding to $\Val(\sigma)$).

Consistently with PyTorch, we perform this computation in \emph{reverse mode} \citep[see \eg][]{rumelhart1986learning},
\ie, for each $n$, we compute $\partial \Val(\sigma)/\partial n$, going from the output node to the input nodes.
This reduces the time complexity by a factor of $|\dhat{E}|$.
We demonstrate the speed-up experimentally.

\paragraph{Optimization loop}

We implemented the optimization in PyTorch, one of the standard tools for differentiable programming and optimization.
Its overview is given in Algo.~\ref{alg:optim}.

Compared to \Regstar, we do not work and update the \emph{strategy} directly, since most of the updates violate the constraint of being a probability distribution.
Instead, we start with a space of unconstrained real-valued \emph{parameters} from which the strategy is generated by
the \emph{Softmax} function. Any update in the parameter space always yields a valid strategy. 

\SetAlFnt{\small}
\begin{algorithm}[h!]
	\SetAlgoLined
	\DontPrintSemicolon
	\SetKwInOut{Parameter}{parameter}\SetKwInOut{Input}{input}\SetKwInOut{Output}{output}
	\SetKwData{C}{C}\SetKwData{D}{D}\SetKwData{MX}{M}\SetKwData{f}{f}
	\SetKw{Or}{or}
	\SetKw{And}{and}
	\SetKw{Not}{not}
	\SetKwData{Array}{array of}
	\SetKwData{MinHeap}{min-heap of Paths}
	\SetKwData{Index}{indexed by}
	\SetKwProg{Macro}{macro}{:}{}
    strategy\_params $\gets$ \textbf{Init}()\;
\For{step $\in$ steps}{ 
strategy $\gets$ \textbf{Softmax}(strategy\_params)\;
	steals = $\{\Steal^\sigma(\hat{e},\tau): \dhat e, \tau\}$ $\gets$ \textbf{Evaluate}(strategy)\;
	loss $\gets$ \textbf{Loss}(steals)\;
	loss.backward()\;
	strategy\_params.grad += \textbf{Noise}(step)\;
	Adam\_optimizer.step()\;
	strategy $\gets$ \textbf{Threshold}(\textbf{Softmax}(strategy\_params))\;
	steals $\gets$ \textbf{Evaluate}(strategy)\;
	dval $\gets$ $\alpha_{\max} - \max(\text{steals})$\;
	\textbf{Save} dval, strategy\;
}
\Return strategy with the highest dval
~\\[2ex]

\caption{Strategy optimization}
	\label{alg:optim}
\end{algorithm}

On the forward pass, $\Steal^\sigma(\hat{e},\tau)$ is evaluated for all $\dhat e$ and $\tau$. It holds that $\AVal_G(\sigma) = \max_{\hat{e},\tau} \Steal^\sigma(\hat{e},\tau)$ \cite[see][Claim~1]{KKMR:Regstar-UAI}, but instead of hard maximum, we use softened variant (denoted by Loss) described below. 
Next, gradients are computed by PyTorch's autodiff, and we add decaying Gaussian noise.
For parameters update, we use Adam optimizer \citep{Adam}.

Note that Softmax never outputs probability distribution containing zeros.
To allow for endpoint values, we cut the outputs at a certain threshold on a test time.
Contrary to \Regstar, we never threshold during optimization, since it disallows using the cut parameters at later stages.
This is crucial for the optimization in changing environment.

\paragraph{Loss function}

To compute the $\Val$, we need to evaluate every $\Steal(\dhat e,\tau)$
and take the maximum $m=\max\{\Steal(\dhat e,\tau)\}$
However, taking $m$ as a loss function leads to much slower
optimization as the signal for parameters update passes through single $\Steal$. Instead,
we take
\begin{equation} \label{E:loss}
	\text{Loss} = \sum_{\dhat e,\tau} \varphi_\varepsilon\bigl(\Steal(\dhat e, \tau)\bigr)^3,
\end{equation}
where 
\begin{itemize}
\item $\varphi_\varepsilon(t)=0$ for $t\in[0,m-\varepsilon)$,
\item $\varphi_\varepsilon(t)=1+(t-m)/\varepsilon$ for $t\in[m-\varepsilon,m]$,
\end{itemize}
where $\varepsilon$ is a hyperparameter.
This choice of loss function optimizes more steals simultaneously and prioritizing
those close to hard maximum $m$.

\paragraph{Strategy initialization}

When optimizing from random strategy, we initialize the parameters so that each outgoing edge gets assigned probability from a uniform distribution on $[0,1)$ which is then normalized over all outgoing edges.

\paragraph{Hyperparameters}

For evaluation on a test time, we \emph{threshold} the probabilities at 0.001 in all experiments.
Remaining hyperparameters, namely optimizer's learning rate (lr), and $\varepsilon$ from in loss~\eqref{E:loss} alter with experiments.

\subsection{Mitigating security holes}
\label{sec-mitigate}

Here we present a proof of Theorem~\ref{thm-hole}. Let $G_1,G_2$ be patrolling graphs and $\sigma_1,\sigma_2$ be
Defender's strategies in $G_1$ and $G_2$, respectively. Assume that:
\begin{enumerate}
\item[(a)] each edge used by $\sigma_1$ is still present in $G_2$;
\item[(b)] for every $(v,m)$ of $G_1$ visited by $\sigma_1$ with positive probability, there is $(v,m')$ of $\sigma_2$ such that $\Val_{G_2}(\sigma_2)(v,m') =  \Val_{G_2}(\sigma_2)$;
\item[(c)] $\hole_{G_1 \mapsto G_2}(\sigma_1,\sigma_1)=0$.
\end{enumerate}
\begin{remark}
The necessity of assumption (a) is apparent: otherwise the Defender could not keep playing according to $\sigma_1$ in $G_2$ at all.
Assumption (b) rules out the possibility that the Defender would switch to $\sigma_2$ in a vertex where $\sigma_2$ is unable to guarantee its long-term
level of protection.
Assumption (c) states that when performing $\sigma_1$, no temporary anomalies arise when the environment changes from $G_1$ and $G_2$,
which would exist neither before, nor after the change. In practice, this condition is scarcely violated---it may happen, \eg, when the Defender
uses edges $e,e'$ with $\tm_1(e)=2,\tm_1(e')=1,\tm_2(e)=1,\tm_2(e')=2$. Then, the total travel time on $e$ and $e'$ is $3$ time units in both $G_1$ and $G_2$.
However, if the environment changes in the middle of the path, then the total travel time increases to $4$.
\end{remark}

Let $\sigma_\kappa$ denote the Defender's strategy which performs the $\kappa$-randomized switch from $\sigma_1$ to $\sigma_2$,
\ie, $\sigma_\kappa$ behaves as $\sigma_1$ in $G_1$, and when the environment changes to $G_2$,
the Defender flips a $\kappa$-biased coin when visiting the next vertex, and switches to $\sigma_2$ only with probability $\kappa$; with the remaining probability $1-\kappa$, it continues executing $\sigma_1$ and flipping the coin in the next vertex again. This goes on until the switch to $\sigma_2$ is performed. 

Now we prove Theorem~\ref{thm-hole}.

\thmhole*

\begin{proof}
The Defender keeps flipping a $\kappa$-biased coin each time he visits a vertex until the $\kappa$-probability comes true.
Hence, the expected number of coin flips (as well as vertices visited) is $1/\kappa$.
Since the time passed between visits to consecutive vertices is at most $\textit{max-time$_2$}$,
the expected time between the environmental change and the switch to $\sigma_2$ is at most $\textit{max-time$_2$}/\kappa$.

As for the security hole, we first prove an upper bound on $\AVal_{G_1 \mapsto G_2}(\sigma_\kappa)$. Thus, fix an Attacker's strategy $\pi$,
a switching time $t\in\Nset$, an initial augmented vertex $\dhat{v}$ and a tuple $(\hat{e},\tau,\delta,t_0)$ such that $\mathbf{P}^{\sigma_\kappa,\pi}(\hat{e},\tau,\delta,t_0)>0$.
We calculate an upper bound on the corresponding steal $s=\Steal^{\sigma_\kappa}(\hat{e},\tau,\delta,t_0)$.
Clearly, exactly one of the following three possibilities occurs:
\begin{enumerate}
	\item the Defender plays according to $\sigma_1$ the entire time the attack is in progress;
	\item the Defender switches from $\sigma_1$ to $\sigma_2$ while the attack is in progress;
	\item the Defender plays according to $\sigma_2$ the entire time the attack is in progress.
\end{enumerate}
Thus, denoting by $p_i$ the probability that the $i$-th possibility occurs
and by $s_i$ the expected cost stolen provided the $i$-th possibility occurs,
we can write $s = p_1s_1 + p_2s_2 + p_3s_3$. Further, the following bounds hold:
\begin{itemize}
	\item $s_1 \leq \max\{\AVal_{G_1}(\sigma_1), \AVal_{G_2}(\sigma_1)\}$ (from assumption (c))
	\item $s_2 \leq \alpha_{\max}(G_2)$ (trivial)
	\item $s_3 \leq \AVal_{G_2}(\sigma_2)$ (from assumption (b))
	\item $p_1 + p_3 \leq 1$ (probabilities of disjoint events)
	\item $p_2 \leq 1 - (1-\kappa)^{\dm}$ (there can be at most $\dm$ coin flips while the attack is in progress)
\end{itemize}
Hence, denoting $\eta_\kappa = (1 - (1-\kappa)^{\dm}) \cdot \alpha_{\max}(G_2)$
and $\Lambda = \max\{\AVal_{G_1}(\sigma_1), \AVal_{G_2}(\sigma_1), \AVal_{G_2}(\sigma_2)\} + \eta_\kappa$,
we get $s \leq \Lambda$.
Since $s$ was the steal for an arbitrary tuple $(\hat{e},\tau,\delta,t_0)$ which can occur with positive probability
and $\EU^{\sigma_\kappa,\pi}(\hat{v})$ is by definition equal to a convex combination of these steals,
it follows that $\EU^{\sigma_\kappa,\pi}(\hat{v}) \leq \Lambda$.
Thus, it is clear from the definition that $\AVal_{G_1 \mapsto G_2}(\sigma_\kappa) \leq \Lambda$.
Subtracting $\max\{\AVal_{G_1}(\sigma_1), \AVal_{G_2}(\sigma_2)\}$ from both sides
and using the trivial fact that
\begin{equation*}
	\max\{a,b,c\}-\max\{a,c\} \leq \max\{0,b-\max\{a,c\}\}
\end{equation*}
holds for any real numbers $a,b,c$, we finally get
\begin{equation*}
\hole_{G_1 \mapsto G_2}(\sigma_\kappa) \leq \varrho + \eta_\kappa.
\end{equation*}
Clearly, $\lim_{\kappa\to 0^+}\eta_\kappa=0$, so the security hole can indeed be pushed arbitrarily close to~$\varrho$.
\end{proof}

\section{Experiments}
\label{app-experiments}

The experiments run on 7 desktop machines with Ubuntu 18.04.5 LTS running on
Intel\textsuperscript{\textregistered} Core\texttrademark{} i7-8700 Processor (6 cores, 12 threads) with 32GB RAM. 
Python version and the required packages are specified in the eclosed pipfile.

\subsection{Strategy improvement analysis}

We assess our strategy synthesis algorithm in comparison with \Regstar on all experiments of \citet{KKMR:Regstar-UAI}. The patrolling graphs considered in \citet{KKMR:Regstar-UAI} model an ATM network in Montreal, and office buildings with $n$-floors connected by stairs. The $n\geq 1$ is a parameter.

The graph for a building with three floors is shown in Fig.~\ref{fig-building}. 
The squares represent offices, and the circles represent corridor locations where the Defender may decide to visit the neighbouring offices. The ``long'' edges represent stairs.
Every office's cost is set to 100.
The time needed to complete an intrusion is set to 100, 200, and 300 for the building with one, two, and three floors, respectively.

\tikzstyle{stoch}=[circle,thick,draw,minimum size=1.5em,inner sep=0em]
\tikzstyle{max}=[rectangle,thick,draw,minimum size=1.5em,inner sep=0em]
\tikzstyle{tran}=[thick,draw,->,>=stealth,rounded corners]
\tikzstyle{loop left}=[tran, to path={.. controls +(150:.8) 
	and +(210:.8) .. (\tikztotarget) \tikztonodes}]
\tikzstyle{loop right}=[tran, to path={.. controls +(30:.8) 
	and +(330:.8) .. (\tikztotarget) \tikztonodes}]
\tikzstyle{loop above}=[tran, to path={.. controls +(60:.5) 
	and +(120:.5) .. (\tikztotarget) \tikztonodes}]
\tikzstyle{loop below}=[tran, to path={.. controls +(240:.8) 
	and +(300:.8) .. (\tikztotarget) \tikztonodes}]
\begin{figure}[tb]
\centering
\begin{tikzpicture}[scale=.5, every node/.style={scale=0.6}, x=1.7cm, y=1.25cm]
	\foreach \x in {1,2,3,4}{
        \foreach \y in {0,2,3,5,6,8}{ 
			\node [max] (r\x\y) at (\x,\y)  {};
        };
        \foreach \d/\y/\u in {0/1/2,3/4/5,6/7/8}{ 
			\node [stoch] (c\x\y) at (\x,\y)  {};
            \draw [tran,-] (c\x\y) -- node[left] {$5$} (r\x\u);
            \draw [tran,-] (c\x\y) -- node[left] {$5$} (r\x\d);
        };
    };
    \foreach \y/\f in {1/{$1^{\textit{st}}$},4/{$2^{\textit{nd}}$},7/{$3^{\textit{rd}}$}}{
        \node [max] (r0\y) at (0,\y)  {};
        \node [max] (r5\y) at (5,\y)  {};
        \draw [tran,-] (r0\y) -- node[above] {$5$} (c1\y);
        \draw [tran,-] (r5\y) -- node[above] {$5$} (c4\y);
        \node at (6,\y) {\f \textrm{ floor}};
        \foreach \x/\r in {1/2,2/3,3/4}{
           \draw [tran,-] (c\x\y) -- node[above] {$2$} (c\r\y);  
        }
    }
    \draw [tran,-,rounded corners] (c11) -- +(-.7,1) -- node[left] {10} +(-.7,2) -- (c14);    
    \draw [tran,-,rounded corners] (c14) -- +(-.7,1) -- node[left] {10} +(-.7,2) -- (c17);    
    \draw [tran,-,rounded corners] (c41) -- +(.7,1) -- node[right] {10} +(.7,2) -- (c44);    
    \draw [tran,-,rounded corners] (c44) -- +(.7,1) -- node[right] {10} +(.7,2) -- (c47);   
\end{tikzpicture}
\caption{A building with three floors.}
\label{fig-building}
\end{figure}
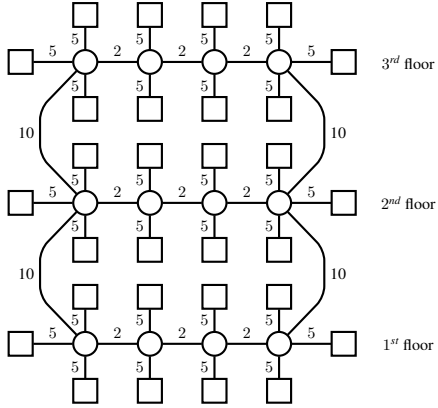

\paragraph{Hyperparameters}

We tested the sensitivity of our optimization scheme to the choice of hyperparameters
on experiments from Sec.~3.5.
For every experiment type (office building with 1, 2 or 3 floors) and memory size (1-8), we run 32 tests (average of 20 optimization trials per 200 steps)
where we sampled learning rate from loguniform distribution over $[0.01, 1]$, $\varepsilon$
from $[0, 0.3]$ uniformly and $\text{pwr}\in\{1,2,3\}$. Here, pwr refers
to the exponent in loss function~\eqref{E:loss}. Outcoming strategy values are normalized by the best ones within category (floor numbers, memory size) and summarized in Fig.~\ref{fig:buildings_grid}.
We observe stable behaviour and we set the final parameters to $\text{lr}=0.15$, $\varepsilon=0.1$ and $\text{pwr}=3$.

For experiment of Sec.~3.4 (Montreal map), we run 20 tests (average of 20 trials per 400 steps)for memory sizes 1-4. We kept $\text{pwr}=3$ and sampled lr and $\varepsilon$ as above.
We observed similar stable behaviour as before and set $\text{lr}=0.1$ and $\varepsilon=0.025$.

\begin{figure}
\centering
	\includegraphics[width=\columnwidth]{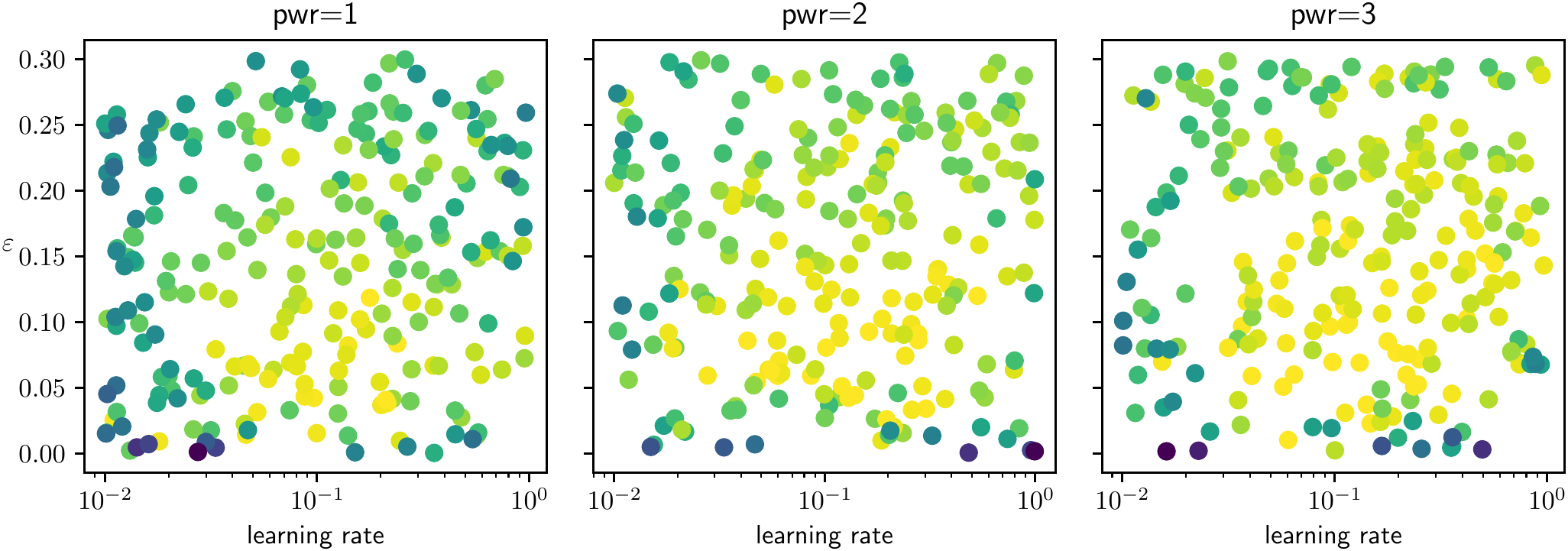}
	\caption{Sensitivity to hyperparameter choice on office building experiment (the brighter color, the higher $\Val$).
	Final parameters were set to $\text{lr}=0.15$, $\varepsilon=0.1$ and $\text{pwr}=3$.}
	\label{fig:buildings_grid}
\end{figure}

\paragraph{Strategy value comparison}

In Fig.~\ref{fig:exp1_montreal} and Fig.~\ref{fig:exp1_building}, we compare the Defender's values obtained by 200 trials of \Regstar and our approach on the examples of Sec.~3.4 and Sec.~3.5 of \citet{KKMR:Regstar-UAI}.  One can see that in all the cases, the new algorithm reaches higher values consistently.

\begin{figure}
\centering
	\includegraphics[width=\columnwidth]{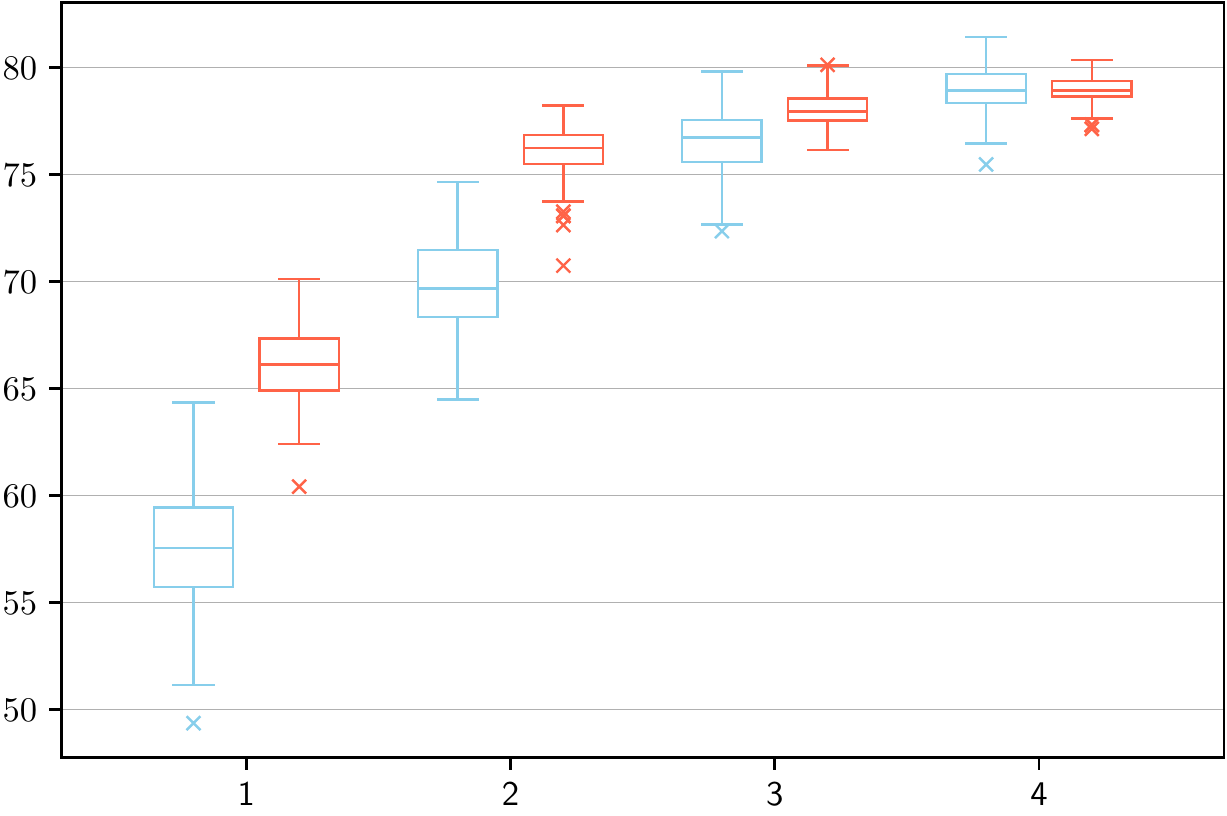}
	\caption{Values of strategies synthesized by \Regstar (blue) and our algorithm (red) on the Montreal-ATM examples with various memory sizes.}
	\label{fig:exp1_montreal}
\end{figure}

\begin{figure}
	\begin{subfigure}{\columnwidth}
		\centering
	    \includegraphics[width=0.9\textwidth]{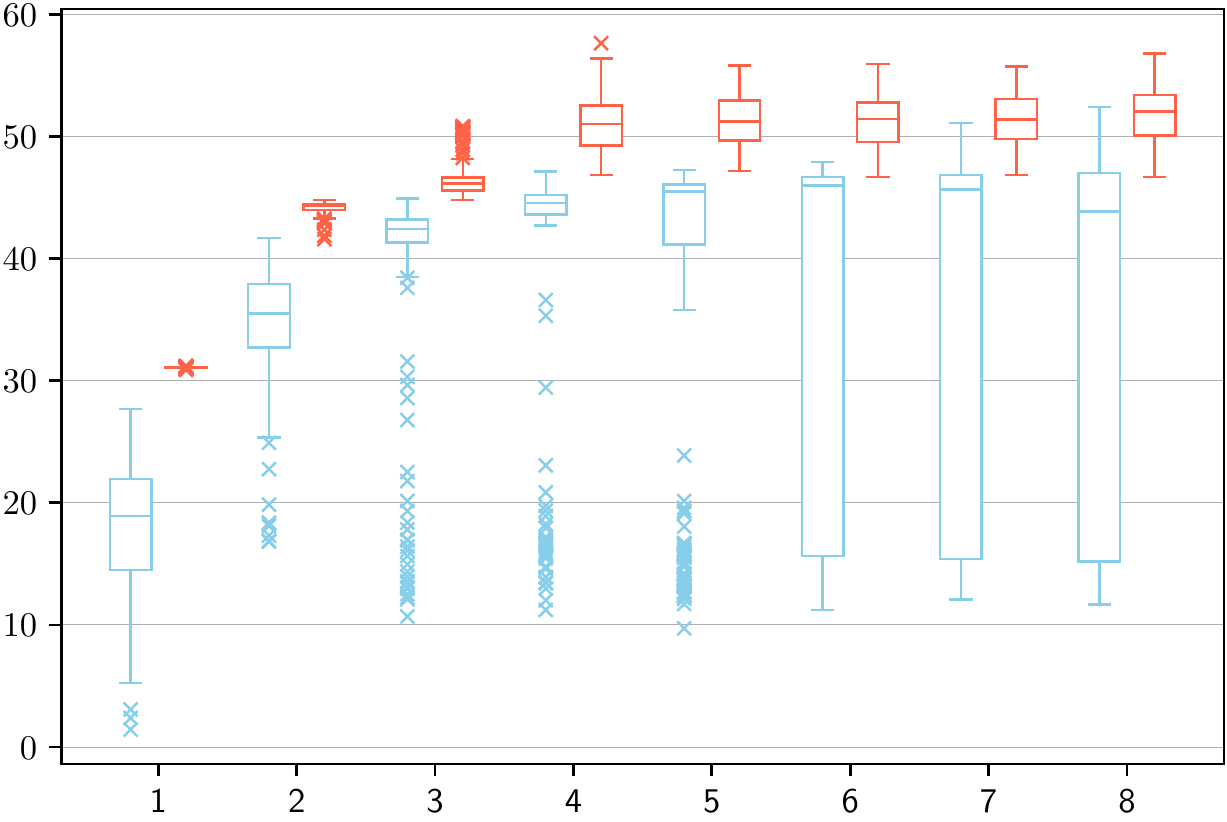}
	    \caption{office building with one floor\\[-2mm]}
	\end{subfigure}
	\begin{subfigure}{\columnwidth}
		\centering\vspace*{2mm}
		\includegraphics[width=0.9\columnwidth]{uai_ex3b_val_compare.pdf}
		\caption{office building with two floors\\[-2mm]}
	\end{subfigure}
	\begin{subfigure}{\columnwidth}
		\centering\vspace*{2mm}
		\includegraphics[width=0.9\columnwidth]{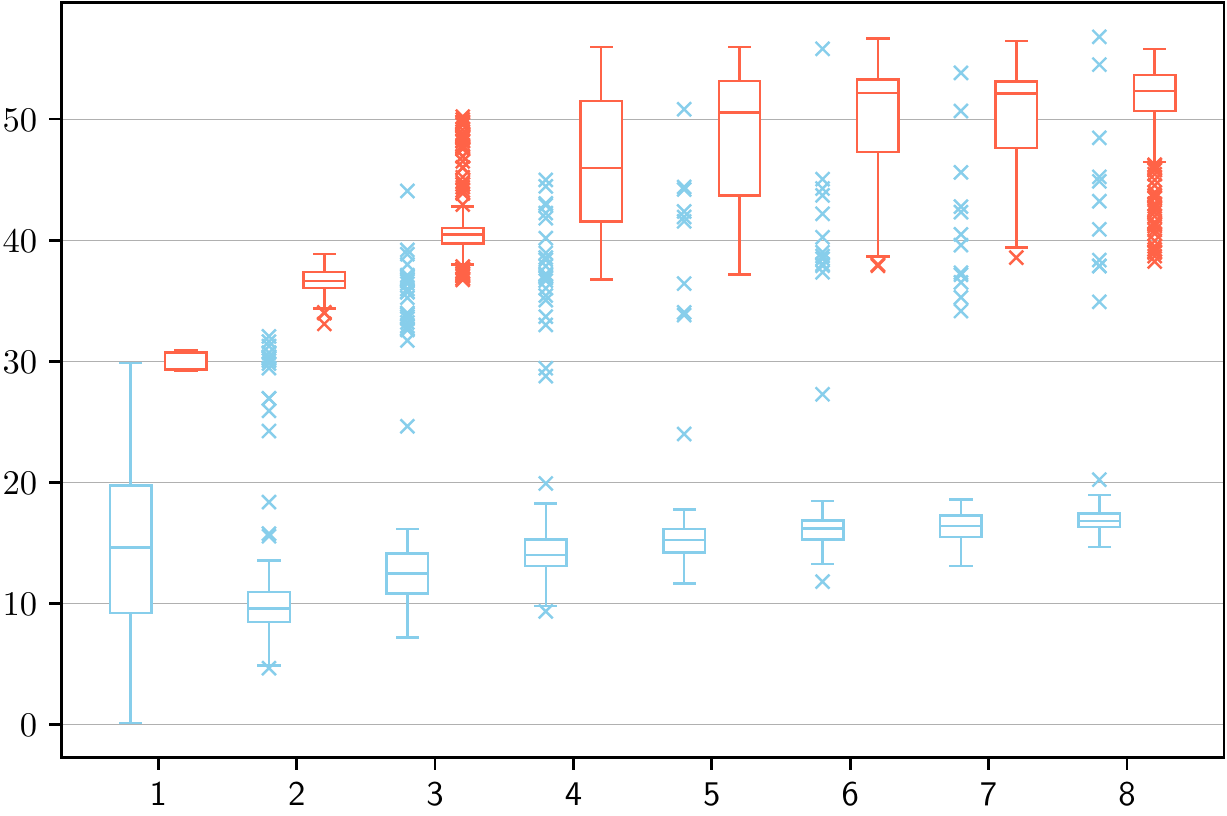}
		\caption{office building with three floors\\[-2mm]}
	\end{subfigure}
	\begin{subfigure}{\columnwidth}
		\centering\vspace*{2mm}
		\includegraphics[width=0.9\columnwidth]{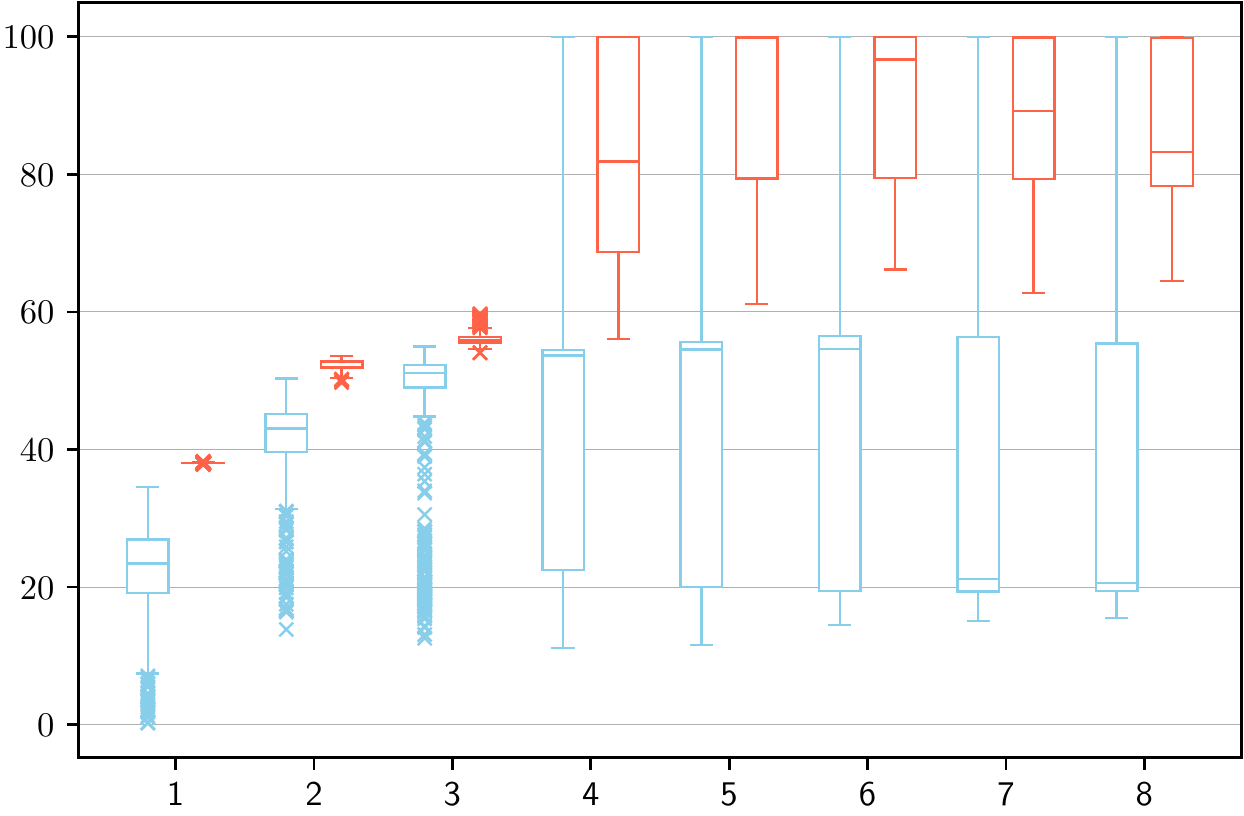}
		\caption{office building with one floor and tight attack time\\[-2mm]}
	\end{subfigure}
	\caption{Values of strategies synthesized by \Regstar (blue) and our algorithm (red) on the building examples with various memory sizes.}
	\label{fig:exp1_building}
\end{figure}

\paragraph{Running time of the forward and backward pass}

In Tab.~\ref{tab:building-times}, we present execution times for all office-buildings experiments.
Note that all the execution times exhibit the same trends as the subpart presented in the main text.

\renewcommand\tabHead{& \multicolumn{2}{c}{forward [ms]}&\multicolumn{2}{c}{backward} \\
$m$ & \multicolumn{1}{c}{\Regstar}& \multicolumn{1}{c}{Ours}& \multicolumn{1}{c}{\Regstar [ms]}& \multicolumn{1}{c}{Ours [ms]} \\
}

\begin{table}
\centering\small
\begin{subtable}{\columnwidth}
		\begin{tabular*}{\columnwidth}{@{}l@{\extracolsep{\fill}}d{2,1}d{2,1}d{4,3}d{1,1}@{}}
            \toprule
            \multicolumn{5}{c}{Office building with one floor}\\[1ex]
\tabHead
			\midrule
			1	& 2 + 0	& 1 + 0	& 3 + 1	& 0 + 0 \\
2	& 6 + 1	& 5 + 1	& 30 + 4	& 1 + 0 \\
3	& 13 + 3	& 11 + 1	& 153 + 18	& 1 + 0 \\
4	& 22 + 2	& 20 + 2	& 473 + 28	& 2 + 0 \\
5	& 34 + 2	& 31 + 3	& 1142 + 72	& 3 + 0 \\
6	& 50 + 3	& 46 + 4	& 2388 + 160	& 4 + 0 \\
7	& 66 + 4	& 62 + 6	& 4419 + 310	& 5 + 0 \\
8	& 85 + 2	& 80 + 8	& 7755 + 561	& 6 + 1 \\
 \end{tabular*}
\end{subtable}
	\begin{subtable}{\columnwidth}
		\begin{tabular*}{\linewidth}{@{}l@{\extracolsep{\fill}}d{3,1}d{3,2}d{4,3}d{2,1}@{}}
\\[1ex]
            \toprule
            \multicolumn{5}{c}{Office building with two floors}\\[1ex]
			\tabHead
			\midrule
			1	& 12 + 1	& 11 + 1	& 44 + 4	& 1 + 0 \\
2	& 50 + 3	& 48 + 5	& 642 + 36	& 4 + 0 \\
3	& 117 + 3	& 110 + 11	& 3210 + 210	& 8 + 1 \\
4	& 217 + 4	& 198 + 19	& 10223 + 725	& 14 + 2 \\
5	& 335 + 7	& 308 + 30	& 26459 + 1920	& 23 + 2 \\
6	& 492 + 4	& 451 + 44	& 57414 + 3906	& 33 + 3 \\
7	& 674 + 8	& 609 + 60	& 106585 + 7044	& 46 + 5 \\
8	& 913 + 9	& 805 + 79	& 186254 + 12028	& 60 + 7 \\
 \end{tabular*}
\end{subtable}
	\begin{subtable}{\columnwidth}
		\begin{tabular*}{\linewidth}{@{}l@{\extracolsep{\fill}}d{3,1}d{3,2}d{4,3}d{2,1}@{}}
\\[1ex]
            \toprule
            \multicolumn{5}{c}{Office building with three floors}\\[1ex]		
			\tabHead
			\midrule
			1	& 37 + 4	& 41 + 4	& 209 + 17	& 3 + 0 \\
2	& 138 + 4	& 176 + 17	& 140 + 73	& 13 + 1 \\
3	& 322 + 18	& 403 + 39	& 846 + 473	& 29 + 3 \\
4	& 649 + 24	& 719 + 70	& 3832 + 2206	& 52 + 5 \\
5	& 1016 + 35	& 1138 + 110	& 10799 + 6012	& 85 + 9 \\
6	& 1471 + 55	& 1631 + 161	& 26630 + 15673	& 122 + 13 \\
7	& 2092 + 86	& 2256 + 220	& 55559 + 31408	& 172 + 19 \\
8	& 2836 + 138	& 2932 + 291	& 101087 + 53157	& 222 + 26 \\
 \end{tabular*}
\end{subtable}
	\begin{subtable}{\columnwidth}
		\begin{tabular*}{\linewidth}{@{}l@{\extracolsep{\fill}}d{3,1}d{3,2}d{4,3}d{2,1}@{}}
\\[1ex]
            \toprule
            \multicolumn{5}{c}{Office building with one floor and tight attack time}\\[1ex]
			\tabHead
			\midrule
			1	& 2 + 0	& 1 + 0	& 2 + 0	& 0 + 0 \\
2	& 5 + 1	& 5 + 0	& 24 + 4	& 1 + 0 \\
3	& 11 + 1	& 12 + 1	& 138 + 9	& 1 + 0 \\
4	& 21 + 6	& 22 + 2	& 441 + 17	& 2 + 0 \\
5	& 32 + 2	& 34 + 3	& 1084 + 46	& 3 + 0 \\
6	& 48 + 13	& 50 + 4	& 2243 + 107	& 4 + 0 \\
7	& 64 + 3	& 68 + 6	& 4183 + 198	& 5 + 0 \\
8	& 86 + 4	& 88 + 7	& 7340 + 345	& 6 + 1 \\
 \end{tabular*}
\end{subtable}
	\caption{Effect of differentiation of $\sigma\mapsto\Val(\sigma)$ in reverse mode (Ours) compared to forward mode (\Regstar) on examples from Sec. 5.3 of \citet{KKMR:Regstar-UAI}.}
	\label{tab:building-times}
\end{table}

\subsection{Changing environment}

We fix a patrolling graph $G_1$ consisting of $15$ locations in the downtown of Vancouver (Fig.~\ref{fig:Vancouver_map}). The target costs are set between $80$ and $100$ at random. Furthermore,
we select $72$ edges connecting the targets with lengths measured in taxicab distance in hundreds of meters. Attack times are fixed to $64$, giving the Defender chance to discover an attack starting $6.4$km far away. For $G_1$, we fix a strategy $\sigma_{1}$ where $\Val_{G_1}(\sigma_1)=42.1$. 

\begin{figure}
	\centering
	\includegraphics[width=.9\columnwidth]{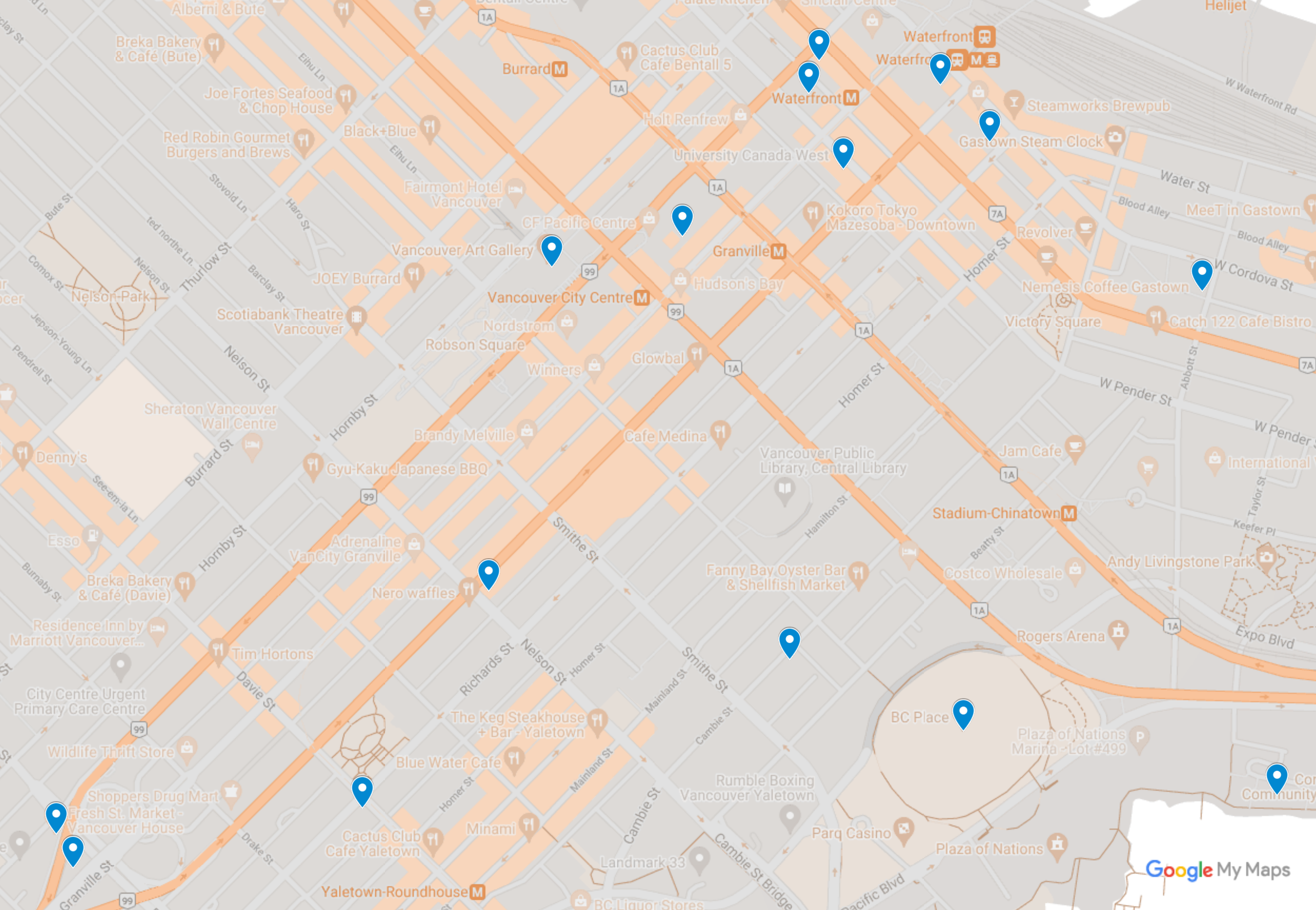}
	\caption{The patrolled locations of Vancouver downtown.}
	\label{fig:Vancouver_map}
\end{figure}

\paragraph{Hyperparameters}

Before making any graph modifications and strategy adaptations, we synthesized
strategy $\sigma_1$ in the original graph. We first tested the sensitivity to hyperparameters
choice by performing 80 tests (average over 20 trials per 600 steps) with lr sampled from loguniform distribution over $[0.01, 1]$ and $\varepsilon$ from loguniform distribution over $[0.001, 0.15]$. Parameter pwr was fixed to 3. The outcomes are summarized in Fig.~\ref{fig:vancouver_grid}.
For all the remaining experiments, we set $\text{lr}=0.07$ and $\varepsilon=0.025$.

\begin{figure}
\centering
	\includegraphics[width=\columnwidth]{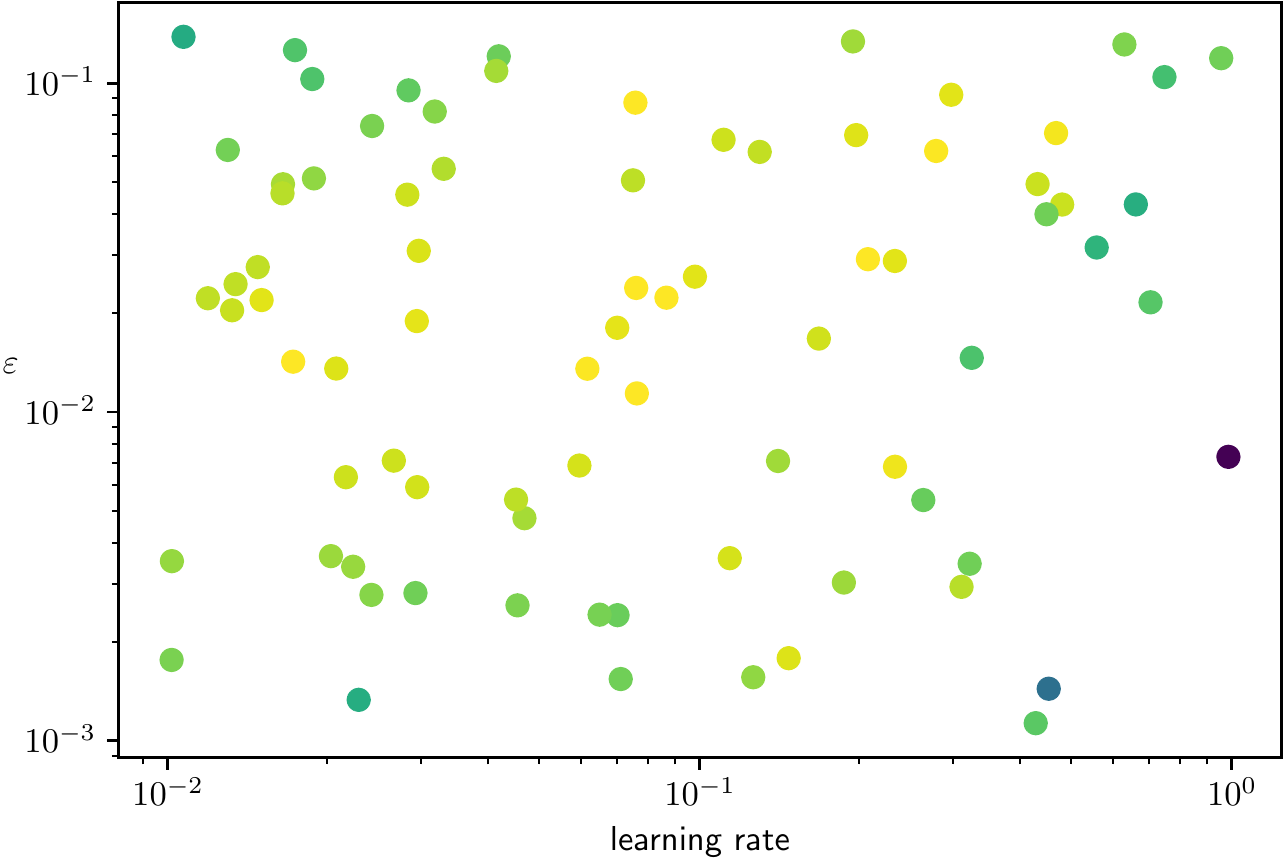}
	\caption{Sensitivity to hyperparameter choice on Vancouver downtown (the brighter color, the higher $\Val$).
	Final parameters were set to $\text{lr}=0.07$, $\varepsilon=0.025$.}
	\label{fig:vancouver_grid}
\end{figure}

\subsection{Detailed results for changing environment}

\paragraph{Utility changes (Table~\ref{tab:variable_cost})}

The cost of each node is increased by its $\changesize\%$ with probability 1/3, decreased by $\changesize\%$ with probability 1/3, or left unchanged. Here, $\changesize$ ranges from $5$ to $30$.
Note that utility changes can modify $\alpha_{\max}$ and thus influence $\Val$.
To compare the values, we normalize all results by $100/\alpha_{\max}$ for each $G_2$ and its~$\alpha_{\max}$.

\let\hc\relax
\definecolor{skyblue}{HTML}{97CCE8}
\definecolor{ourorange}{HTML}{ED6D52}
\definecolor{ourgreen}{HTML}{73FA79}
\renewcommand\hcglob[3]{\cellcolor{skyblue!\fpeval{((#1-#2)/(#3-#2))**2*100}!ourorange}#1}
\renewcommand\hcglobg[3]{\cellcolor{ourorange!\fpeval{((#1-#2)/(#3-#2))*100}!ourgreen}#1}

\begin{table}
	\small
	\centering
	\caption{Variable cost
	}
	\begin{adjustbox}{max width=\linewidth}
		\begin{tabular*}{\linewidth}{@{}l@{\extracolsep{\fill}}rd{3,2}d{3,2}d{3,2}d{3,2}}
			\toprule
			\multirow{2}{*}{$\changesize$} &\multirow{2}{*}{steps}& \multicolumn{2}{c}{ $\Val$} & \multicolumn{2}{c}{Security Hole} 
			\\
			&  & \multicolumn{1}{c}{from $\sigma_{1}$} & \multicolumn{1}{c}{from rnd} & \multicolumn{1}{c}{\phantom{x}from $\sigma_{1}$} & \multicolumn{1}{c}{from rnd}
			\\
			\midrule
			\multirow{5}{*}{5}
			\gdef\hc#1{\hcglob{#1}{12.7}{43.8}}
			\gdef\hcg#1{\hcglobg{#1}{0}{33.7}}
			& 0		& \hc{40.9}	+ 0.9	& \hc{12.7}	+ 3.2	& \hcg{0.0}	+ 0.0	& \hcg{ 4.6}	+ 2.7	\\
			& 50	& \hc{43.4}	+ 0.6	& \hc{27.4}	+ 0.8	& \hcg{2.8}	+ 1.2	& \hcg{14.4}	+ 1.2	\\
			& 100	& \hc{43.6}	+ 0.5	& \hc{37.3}	+ 1.2	& \hcg{3.9}	+ 1.9	& \hcg{25.3}	+ 2.6	\\
			& 200	& \hc{43.8}	+ 0.6	& \hc{41.3}	+ 0.6	& \hcg{5.2}	+ 2.8	& \hcg{30.2}	+ 4.4	\\
			& 400	& \hc{43.8}	+ 0.6	& \hc{42.7}	+ 0.5	& \hcg{6.9}	+ 4.4	& \hcg{33.7}	+ 3.9	\\
			\midrule                                                               
			\multirow{5}{*}{10}
			\gdef\hc#1{\hcglob{#1}{14.3}{46}}
			\gdef\hcg#1{\hcglobg{#1}{0}{33.1}}
			& 0		& \hc{40.9}	+ 0.9	& \hc{14.3}	+ 3.8	& \hcg{0.0}	+ 0.0	& \hcg{ 5.8}	+ 2.3	\\
			& 50	& \hc{45.3}	+ 0.8	& \hc{30.3}	+ 1.4	& \hcg{5.3}	+ 2.4	& \hcg{15.5}	+ 2.0	\\
			& 100	& \hc{45.7}	+ 0.9	& \hc{39.6}	+ 1.4	& \hcg{7.1}	+ 3.3	& \hcg{27.3}	+ 3.9	\\
			& 200	& \hc{45.9}	+ 1.0	& \hc{44.0}	+ 0.9	& \hcg{8.4}	+ 4.2	& \hcg{30.0}	+ 4.6	\\
			& 400	& \hc{46.0}	+ 1.0	& \hc{45.4}	+ 0.8	& \hcg{9.2}	+ 4.7	& \hcg{33.1}	+ 3.5	\\
			\midrule
			\multirow{5}{*}{20}
			\gdef\hc#1{\hcglob{#1}{14.4}{49.9}}
			\gdef\hcg#1{\hcglobg{#1}{0}{35.6}}
			& 0		& \hc{40.9}	+ 0.9	& \hc{14.4}	+ 4.0	& \hcg{ 0.0}	+ 0.0	& \hcg {5.4}	+ 1.9	\\
			& 50	& \hc{48.2}	+ 2.0	& \hc{35.4}	+ 2.8	& \hcg{ 7.7}	+ 1.4	& \hcg{19.0}	+ 2.8	\\
			& 100	& \hc{48.9}	+ 2.3	& \hc{44.0}	+ 2.2	& \hcg{ 9.3}	+ 2.6	& \hcg{27.3}	+ 4.3	\\
			& 200	& \hc{49.5}	+ 2.3	& \hc{48.4}	+ 2.4	& \hcg{10.3}	+ 3.4	& \hcg{33.5}	+ 4.5	\\
			& 400	& \hc{49.9}	+ 2.3	& \hc{49.8}	+ 2.2	& \hcg{12.4}	+ 3.2	& \hcg{35.6}	+ 3.9	\\
			\midrule
			\multirow{5}{*}{30}
			\gdef\hc#1{\hcglob{#1}{14.7}{54}}
			\gdef\hcg#1{\hcglobg{#1}{0}{40}}
			& 0		& \hc{40.9}	+ 0.9	& \hc{14.7}	+ 4.2	& \hcg{ 0.0}	+ 0.0	& \hcg{ 6.1}	+ 2.2	\\
			& 50	& \hc{50.4}	+ 3.4	& \hc{38.2}	+ 4.4	& \hcg{ 9.7}	+ 1.9	& \hcg{19.6}	+ 3.3	\\
			& 100	& \hc{51.4}	+ 3.6	& \hc{48.7}	+ 2.8	& \hcg{11.3}	+ 3.1	& \hcg{30.1}	+ 5.6	\\
			& 200	& \hc{52.4}	+ 3.8	& \hc{52.8}	+ 3.5	& \hcg{14.0}	+ 2.3	& \hcg{39.4}	+ 4.1	\\
			& 400	& \hc{52.9}	+ 3.7	& \hc{54.0}	+ 3.3	& \hcg{15.2}	+ 3.2	& \hcg{40.0}	+ 5.6	\\
			\bottomrule
		\end{tabular*}
	\end{adjustbox}
	\label{tab:variable_cost}
\end{table}

\paragraph{Variable edge length (Table~\ref{tab:variable_length})}

As in the previous case, the length of each edge is increased/decreased by $\changesize\%$ or kept unchanged (with the same probability).

\begin{table}
	\small
	\centering
	\caption{Variable length
	}
	\begin{adjustbox}{max width=\linewidth}
		\begin{tabular*}{\linewidth}{@{}l@{\extracolsep{\fill}}rd{3,3}d{3,2}d{3,2}d{3,2}}
			\toprule
			\multirow{2}{*}{$\changesize$} & \multirow{2}{*}{steps} & \multicolumn{2}{c}{ $\Val$ } & \multicolumn{2}{c}{Security Hole} 
			\\
			& & \multicolumn{1}{c}{from $\sigma_{1}$} & \multicolumn{1}{c}{from rnd} & \multicolumn{1}{c}{from $\sigma_{1}$} & \multicolumn{1}{c}{from rnd}
			\\
			\midrule\multirow{5}{*}{5}
			\gdef\hc#1{\hcglob{#1}{9.4}{41}}
			\gdef\hcg#1{\hcglobg{#1}{0.6}{33.8}}
			& 0     & \hc{36.4} + 2.2    & \hc{9.4} + 0.3     & \hcg{0.6} + 1.0     & \hcg{2.9} + 0.3     \\
			& 50    & \hc{40.3} + 0.9    & \hc{24.3} + 0.6    & \hcg{5.0} + 1.9     & \hcg{13.0} + 1.1    \\
			& 100   & \hc{40.6} + 0.9    & \hc{34.7} + 0.7    & \hcg{6.5} + 2.6     & \hcg{26.2} + 2.5    \\
			& 200   & \hc{40.9} + 0.9    & \hc{39.3} + 0.7    & \hcg{8.8} + 3.1     & \hcg{31.9} + 2.3    \\
			& 400   & \hc{41.0} + 0.9    & \hc{40.4} + 0.8    & \hcg{9.4} + 3.1     & \hcg{33.8} + 2.7    \\
			\midrule\multirow{5}{*}{10}          
			\gdef\hcg#1{\hcglobg{#1}{0.1}{34.7}}
			& 0     & \hc{33.1} + 4.2    & \hc{9.4} + 0.5     & \hcg{0.1} + 0.2     & \hcg{2.9} + 0.7     \\
			& 50    & \hc{39.3} + 1.8    & \hc{24.5} + 1.0    & \hcg{6.5} + 2.7     & \hcg{12.7} + 1.3    \\
			& 100   & \hc{40.0} + 1.5    & \hc{35.4} + 1.0    & \hcg{9.5} + 2.5     & \hcg{25.6} + 1.9    \\
			& 200   & \hc{40.6} + 1.5    & \hc{39.8} + 1.0    & \hcg{9.9} + 3.3     & \hcg{32.3} + 2.3    \\
			& 400   & \hc{41.0} + 1.6    & \hc{40.9} + 1.2    & \hcg{10.7} + 3.3    & \hcg{34.7} + 1.7    \\
			\midrule\multirow{5}{*}{20}          
			\gdef\hc#1{\hcglob{#1}{8.9}{42.4}}
			\gdef\hcg#1{\hcglobg{#1}{0}{34.3}}
			& 0     & \hc{20.9} + 10.1   & \hc{8.9} + 0.7     & \hcg{0.0} + 0.1     & \hcg{2.3} + 1.0     \\
			& 50    & \hc{36.1} + 3.7    & \hc{24.6} + 1.7    & \hcg{7.6} + 5.9     & \hcg{13.3} + 1.4    \\
			& 100   & \hc{37.9} + 2.9    & \hc{36.5} + 1.7    & \hcg{8.6} + 4.1     & \hcg{28.1} + 2.0    \\
			& 200   & \hc{39.1} + 2.3    & \hc{41.0} + 1.7    & \hcg{12.9} + 4.1    & \hcg{32.7} + 2.2    \\
			& 400   & \hc{40.4} + 2.1    & \hc{42.4} + 1.9    & \hcg{14.2} + 3.7    & \hcg{34.3} + 3.4    \\
			\midrule\multirow{5}{*}{30}          
			\gdef\hc#1{\hcglob{#1}{8.1}{44.6}}
			\gdef\hcg#1{\hcglobg{#1}{0.1}{34.1}}
			& 0     & \hc{16.3} + 10.1   & \hc{8.1} + 1.1     & \hcg{0.1} + 0.4     & \hcg{1.7} + 1.2     \\
			& 50    & \hc{33.6} + 5.0    & \hc{24.0} + 2.5    & \hcg{5.4} + 5.4     & \hcg{13.2} + 2.8    \\
			& 100   & \hc{36.4} + 3.7    & \hc{37.8} + 2.4    & \hcg{10.2} + 4.6    & \hcg{27.5} + 2.3    \\
			& 200   & \hc{38.7} + 3.2    & \hc{42.7} + 2.2    & \hcg{14.2} + 5.3    & \hcg{33.9} + 2.6    \\
			& 400   & \hc{40.8} + 2.6    & \hc{44.4} + 2.9    & \hcg{20.9} + 3.8    & \hcg{34.1} + 3.5    \\
			\bottomrule                            
		\end{tabular*}
	\end{adjustbox}
	\label{tab:variable_length}
\end{table}

\paragraph{Removed edges (Table~\ref{tab:removed_edges})}

We randomly delete $\changesize=1,2,4,8$ edges so that $G_2$ remains strongly connected.

\begin{table}
	\small
	\centering
	\caption{Removed edges
	}
	\begin{adjustbox}{max width=\linewidth}
		\begin{tabular*}{\linewidth}{@{}l@{\extracolsep{\fill}}rd{3,3}d{3,2}d{3,2}d{3,2}}
			\toprule
			\multirow{2}{*}{$\changesize$} & \multirow{2}{*}{steps}& \multicolumn{2}{c}{ $\Val$ } & \multicolumn{2}{c}{Security Hole} 
			\\
			& & \multicolumn{1}{c}{from $\sigma_{1}$} & \multicolumn{1}{c}{from rnd} & \multicolumn{1}{c}{from $\sigma_{1}$} & \multicolumn{1}{c}{from rnd}
			\\
			\midrule\multirow{5}{*}{1}
			\gdef\hc#1{\hcglob{#1}{9.6}{42.1}}
			\gdef\hcg#1{\hcglobg{#1}{0}{34.5}}
			& 0     & \hc{39.5} + 5.3    & \hc{9.6} + 0.3     & \hcg{0.0} + 0.0     & \hcg{3.1} + 0.4     \\
			& 50    & \hc{42.0} + 0.2    & \hc{24.6} + 0.3    & \hcg{0.9} + 1.5     & \hcg{13.3} + 0.8    \\
			& 100   & \hc{42.1} + 0.1    & \hc{35.1} + 0.7    & \hcg{0.9} + 1.5     & \hcg{24.9} + 2.1    \\
			& 200   & \hc{42.1} + 0.1    & \hc{39.4} + 0.3    & \hcg{0.9} + 1.4     & \hcg{31.9} + 3.3    \\
			& 400   & \hc{42.1} + 0.1    & \hc{40.9} + 0.5    & \hcg{1.3} + 1.2     & \hcg{34.5} + 2.3    \\
			\midrule\multirow{5}{*}{2}
			\gdef\hcg#1{\hcglobg{#1}{0}{35.7}}
			& 0     & \hc{37.1} + 5.5    & \hc{9.6} + 0.3     & \hcg{0.0} + 0.0     & \hcg{3.0} + 0.4     \\
			& 50    & \hc{41.8} + 0.3    & \hc{24.7} + 0.3    & \hcg{1.6} + 1.8     & \hcg{12.8} + 1.4    \\
			& 100   & \hc{41.9} + 0.2    & \hc{34.9} + 0.5    & \hcg{1.6} + 1.8     & \hcg{24.4} + 3.1    \\
			& 200   & \hc{42.0} + 0.2    & \hc{39.6} + 0.3    & \hcg{1.8} + 1.6     & \hcg{32.8} + 2.2    \\
			& 400   & \hc{42.1} + 0.2    & \hc{41.0} + 0.2    & \hcg{2.2} + 1.4     & \hcg{35.7} + 2.1    \\
			\midrule\multirow{5}{*}{4}
			\gdef\hc#1{\hcglob{#1}{9.8}{41.7}}
			\gdef\hcg#1{\hcglobg{#1}{0.1}{36.8}}
			& 0     & \hc{28.0} + 9.8    & \hc{9.8} + 0.3     & \hcg{0.1} + 0.5     & \hcg{3.4} + 0.4     \\
			& 50    & \hc{40.4} + 1.7    & \hc{24.9} + 0.4    & \hcg{4.4} + 3.1     & \hcg{13.3} + 1.2    \\
			& 100   & \hc{41.1} + 1.1    & \hc{35.2} + 0.6    & \hcg{5.9} + 4.1     & \hcg{26.2} + 1.5    \\
			& 200   & \hc{41.5} + 0.6    & \hc{39.6} + 0.8    & \hcg{6.2} + 4.3     & \hcg{31.8} + 4.1    \\
			& 400   & \hc{41.7} + 0.4    & \hc{40.8} + 0.3    & \hcg{6.6} + 3.7     & \hcg{36.8} + 2.5    \\
			\midrule\multirow{5}{*}{8}
			\gdef\hc#1{\hcglob{#1}{9.8}{40.9}}
			\gdef\hcg#1{\hcglobg{#1}{0}{34.4}}
			& 0     & \hc{17.0} + 11.4   & \hc{9.8} + 0.9     & \hcg{0.0} + 0.0     & \hcg{3.3} + 0.6     \\
			& 50    & \hc{36.1} + 9.1    & \hc{24.8} + 0.3    & \hcg{8.2} + 5.3     & \hcg{13.2} + 1.7    \\
			& 100   & \hc{39.7} + 1.9    & \hc{35.3} + 0.9    & \hcg{10.3} + 5.1    & \hcg{25.5} + 2.4    \\
			& 200   & \hc{40.4} + 1.4    & \hc{39.0} + 0.7    & \hcg{10.9} + 5.1    & \hcg{29.8} + 2.1    \\
			& 400   & \hc{40.9} + 1.0    & \hc{40.5} + 0.7    & \hcg{11.8} + 5.5    & \hcg{34.4} + 2.6    \\
			\bottomrule
		\end{tabular*}
	\end{adjustbox}
	\label{tab:removed_edges}
\end{table}

\end{document}